\documentclass[journal]{IEEEtran}

\usepackage{amssymb}
\usepackage{amsmath}
\usepackage{amsfonts}
\usepackage{multirow}
\usepackage{graphicx}
\usepackage{textcomp}
\usepackage{amsthm}
\renewcommand{\vec}[1]{\boldsymbol{#1}}
\usepackage{setspace}
\usepackage{subfigure}
\usepackage{algorithmic}
\usepackage[ruled,vlined,boxed,linesnumbered]{algorithm2e}
\usepackage[dvipsnames]{xcolor}
\usepackage{bm}
\usepackage{url}

\allowdisplaybreaks[4]

\ifCLASSINFOpdf
\else
\fi

\newtheorem{thm}{Theorem}
\newtheorem{lem}{Lemma}

\newtheorem{defn}{Definition}

\newtheorem{pro}{Problem}

\begin{document}

\title{Multi-Task Diffusion Incentive Design for Mobile Crowdsourcing in Social Networks}

\author{Jianxiong Guo,~\IEEEmembership{Member,~IEEE},
	Qiufen Ni,
	Weili Wu,~\IEEEmembership{Senior Member,~IEEE},
	and Ding-Zhu Du
	\thanks{Jianxiong Guo is with the Advanced Institute of Natural Sciences, Beijing Normal University, Zhuhai 519087, China, and also with the Guangdong Key Lab of AI and Multi-Modal Data Processing, BNU-HKBU United International College, Zhuhai 519087, China. (E-mail: jianxiongguo@bnu.edu.cn)
	
	Qiufen Ni is with the School of Computers, Guangdong University of Technology, Guangzhou 510006, China. (E-mail: niqiufen@gdut.edu.cn)
		
	Weili Wu and Ding-Zhu Du are with the Department of Computer Science, Erik Jonsson School of Engineering and Computer Science, The University of Texas at Dallas, Richardson, TX 75080, USA. (E-mail: weiliwu@utdallas.edu; dzdu@utdallas.edu)
	
	\textit{(Corresponding author: Qiufen Ni.)}
	}% <-this 
	\thanks{Manuscript received xxxx; revised xxxx.}}

\markboth{Journal of \LaTeX\ Class Files,~Vol.~xx, No.~xx, March~2023}%
{Shell \MakeLowercase{\textit{et al.}}: Bare Demo of IEEEtran.cls for IEEE Journals}

\maketitle

\begin{abstract}
    Mobile Crowdsourcing (MCS) is a novel distributed computing paradigm that recruits skilled workers to perform location-dependent tasks. A number of mature incentive mechanisms have been proposed to address the worker recruitment problem in MCS systems. However, they all assume that there is a large enough worker pool and a sufficient number of users can be selected. This may be impossible in large-scale crowdsourcing environments. To address this challenge, we consider the MCS system defined on a location-aware social network provided by a social platform. In this system, we can recruit a small number of seed workers from the existing worker pool to spread the information of multiple tasks in the social network, thus attracting more users to perform tasks. In this paper, we propose a Multi-Task Diffusion Maximization (MT-DM) problem that aims to maximize the total utility of performing multiple crowdsourcing tasks under the budget. To accommodate multiple tasks diffusion over a social network, we create a multi-task diffusion model, and based on this model, we design an auction-based incentive mechanism, MT-DM-L. To deal with the high complexity of computing the multi-task diffusion, we adopt Multi-Task Reverse Reachable (MT-RR) sets to approximate the utility of information diffusion efficiently. Through both complete theoretical analysis and extensive simulations by using real-world datasets, we validate that our estimation for the spread of multi-task diffusion is accurate and the proposed mechanism achieves individual rationality, truthfulness, computational efficiency, and $(1-1/\sqrt{e}-\varepsilon)$ approximation with at least $1-\delta$ probability.
\end{abstract}


\IEEEpeerreviewmaketitle

\section{Introduction}
\IEEEPARstart{W}{ith} the rapid progress of mobile devices and wireless communication technologies in the last two decades, Mobile Crowdsourcing (MCS), as a new distributed computing paradigm, has attracted attention from academia and industry because of its wide applications and great commercial value. In many fields, such as parking service \cite{hoh2012trucentive}, traffic report \cite{philip2022traffic}, air monitoring \cite{li2019selecting}, surveillance video analysis \cite{tahboub2015intelligent}, and image collection \cite{zhang2018crowdbuy}, they can all be summarized into a typical pattern that utilizes the power of the crowd to perform location-dependent tasks. This is usually composed of three parts: task requesters, participant workers, and the MCS platform. The MCS platform first publishes these tasks that are collected from task requesters through the crowdsourcing platform, then recruits a number of workers to perform their assigned tasks and pays them as an incentive for participation.

Incentive mechanisms are imperative for an MCS system because recruited workers have costs to perform tasks. There is a lot of research focused on designing incentive mechanisms \cite{feng2014trac} \cite{zhang2015truthful} \cite{jiang2019incentivizing} to motivate more workers to participate in crowdsourcing. They are all based on a basic assumption that the number of participants is enough to accomplish tasks. However, the MCS platform may face the cold-start problem where there are not enough participants at the early stage of platform operation or the task is relatively unpopular, resulting in not many users being interested in this task. Here, most existing incentive mechanisms aim to motivate as many users as possible in the existing worker pool to participate instead of expanding the worker pool, which cannot effectively solve the bottleneck. Thus, how to expand the worker pool by diffusing the tasks is not a trivial issue.

Fortunately, Online Social Networks (OSNs) have been greatly developed in the past ten years, which provide an effective solution for diffusing news, innovation, and product adoption. Based on the information diffusion, the Influence Maximization (IM) problem \cite{kempe2003maximizing} has been formulated, which aims to select a size-$k$ seed set such that the influence spread can be maximized over the social network. A natural idea is to develop the MCS platform based on OSNs so that more users can participate by spreading tasks. For example, Meituan (a take-out app popular in China) is deeply bound to WeChat so that it can use the social network of WeChat to diffuse take-out orders, thus alleviating the shortage of take-out staff during the peak meal period. Other examples, such as Steps on Facebook \cite{steps}, Image Labeler on Google+ \cite{imageLabeler}, and Mechanical Turk on Amazon \cite{turk}, are all MCS platforms which leverage the knowledge provided by online social platforms. 
%In addition to expanding the worker pool, it can also promote high-quality task execution because of peer pressure while beneficial to protect users' privacy.

However, the integration of the MCS platform and social network will bring new technical challenges. These problems should be solved before designing an incentive mechanism for the MCS system on social networks. First, in the traditional IM problem, there is only one cascade diffusing over the network, but the MCS problem needs to publish multiple types of crowdsourcing tasks at the same time. Thus, how to model the diffusion process of multiple tasks throughout OSNs and define the utility of task diffusion are two crucial problems in this scenario. Second, it is impossible for the MCS platform to invite every registered participant in the worker pool due to the requirements for the limited budget, low intrusiveness, and execution quality. Here, the cost of completing the assigned tasks and diffusing tasks is neglected in the current research. Third, in the traditional IM problem, it is \#P-hard to compute the influence spread given a seed set under the simplest Independent Cascade (IC) model \cite{chen2010scalable}. The calculation of the utility of multiple task diffusions is much more complex than the IC model. Thus, we need to estimate this utility with high efficiency, which is important to guarantee the computational efficiency of the incentive mechanism. Thus, in this paper, we propose a Multi-Task Diffusion Maximization (MT-DM) problem, which aims to design a truthful incentive mechanism to select a winning seed set from registered users in the worker pool under the limited budget to diffuse multiple tasks and maximize the total utility of performing crowdsourcing tasks. The workflow of our MT-DM problem is shown in Fig. \ref{fig1}.

\begin{figure}[!t]
	\centering
	\includegraphics[width=\linewidth]{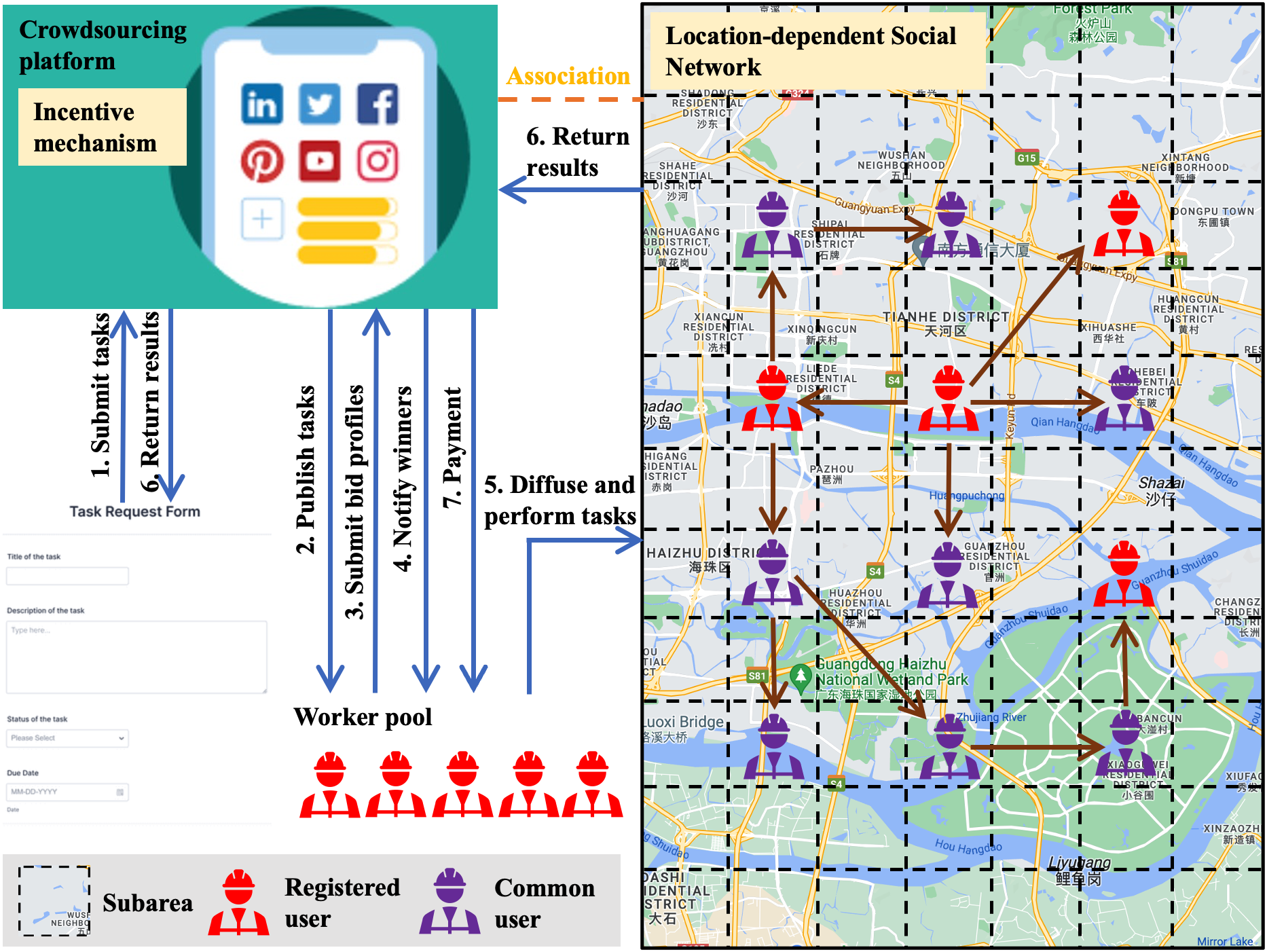}
	\caption{The workflow of incentive in the MT-DM problem.}
	\label{fig1}
\end{figure}

For the MT-DM problem, we first take a sealed reverse auction, MT-DM-L, as the incentive mechanism. Each registered user in the worker pool gives it bidding information for the tasks it can perform and diffuse. The MCS platform selects a certain number of registered users under a given budget and notifies those winning workers. The winning workers begin to diffuse their assigned tasks based on the underlying location-dependent social network. Because a winning worker can undertake multiple tasks, we then design a Multi-Task IC model to simulate the diffusion process of multiple takes on the same network. Based on the Multi-Task IC model, we define the objective of our MT-DM problem as a Tasks Diffusion Function. We prove that the MT-DM problem is NP-hard and the Tasks Diffusion Function is monotone and submodular. Through theoretical analysis, the MT-DM-L satisfies individual rationality, truthfulness, and computational efficiency, which will avoid price manipulation and guarantees a fair and competitive market environment. However, computing the tasks diffusion given a seed worker set is \#P-hard, and there is no effective method to estimate the marginal gain of tasks diffusion if using the greedy strategy in MT-DM-L. In order to reduce the time complexity, motivated by the technique of reverse influence sampling (RIS) \cite{borgs2014maximizing} to address the IM problem, we design a new sampling method based on random Multi-Task Reverse Reachable (MT-RR) sets adapted to the Multi-Task IC model. Next, we design an efficient approximation algorithm, Modified-OPIM-C, based on our new sampling method. Through integrating the MT-DM-L with Modified-OPIM-C, the marginal gain of tasks diffusion can be efficiently estimated, and it guarantees to return a solution with $(1-1/\sqrt{e}-\varepsilon)$ with at least $(1-\delta)$ probability. Finally, we conduct an extensive evaluation by using real-world datasets. By comparing our incentive mechanism with the-state-of-art baselines, the experimental results validate the effectiveness of our proposed sampling and algorithm in approximate performance and efficiency.

\textbf{Organization.} In Section \uppercase\expandafter{\romannumeral2}, we summarize the related works. We then introduce our MT-DM problem and its basic properties in Section \uppercase\expandafter{\romannumeral3}, and incentive mechanism in Section \uppercase\expandafter{\romannumeral4}. In Section \uppercase\expandafter{\romannumeral5} and \uppercase\expandafter{\romannumeral6}, we elaborate on the sampling-based algorithm, its corresponding theoretical analysis, and workflow. Experiments and discussions are presented in Section \uppercase\expandafter{\romannumeral7}, and finally, Section \uppercase\expandafter{\romannumeral8} concludes this paper.

\section{Related Work}
In this section, we provide a brief review of the related research on the following topics.

\textit{Information Diffusion in Social Networks.} Through the ``word-of-mouth'' effect among neighboring users in social networks, Kempe \textit{et al.} \cite{kempe2003maximizing} first summarized it as a combinatorial optimization problem, Influence Maximization, and proposed two classical diffusion models, IC model and Linear Threshold (LT) model. Due to its NP-hardness \cite{kempe2003maximizing} and the \#P-hardness to compute the influence spread \cite{chen2010scalable}, a series of researchers continued to study low-complexity methods to get near-optimal solutions in an acceptable time. Brogs \textit{et al.} created the technique of reverse influence sampling (RIS) and defined the concept of random reverse reachable (RR) sets, which can be used to unbiasedly estimate the influence spread. Based on that, a plethora of works followed the RIS to further improve the efficiency of IM problem, such as TIM/TIM+ \cite{tang2014influence}, IMM \cite{tang2015influence}, SSA/D-SSA \cite{nguyen2016stop}, OPIM-C \cite{tang2018online}, and HIST \cite{guo2020influence}. They gradually reduced the number of random RR sets or restricted the size of a RR set by utilizing probability analysis, which can achieve a $(1-1/e-\varepsilon)$ approximation with high probability in near-linear time. Even though the worker recruitment scenario in social networks has been studied in the IM problem, it is only required to diffuse information, which is different from ours. Our goal is to design an incentive mechanism to identify workers to perform and diffuse tasks, rather than maximizing the influence spread as much as possible.

\textit{Worker Recruitment in MCS.} Many incentive mechanisms for crowdsourcing have been proposed so far, where the worker recruitment problem \cite{tong2020spatial} is a representative research direction and the auction theory has been used as an effective method to address users' strategic behaviors in MCS. Yang \textit{et al.} \cite{yang2012crowdsourcing} was the first to take an auction model as an incentive mechanism in the crowdsourcing for a mobile phone sensing system. Feng \textit{et al.} \cite{feng2014trac} presented a truthful auction to incentive location-aware collaborative sensing in MCS with polynomial time and near-optimal task allocation. Zhou \textit{et al.} \cite{zhou2017incentive} considered mobile crowdsensing systems under the budget and capacity constraint, which can not only protect users' privacy but also achieve near-optimal benefits. To maximize social welfare in MCS, Wang \textit{et al.} \cite{wang2018truthful} proposed a location-aware privacy-preserving two-stage auction algorithm and Gao \textit{et al.} \cite{gao2018truthful} designed a reverse auction model with a nearly minimum social cost. To maximize the total utility of MCS platforms, Guo \textit{et al.} \cite{guo2021reliable} \cite{guo2022theoretical} put forward two kinds of budgeted MCS problems in Internet-of-Vehicles and Cloud-Edge environments. However, the above-mentioned works all studied worker recruitment problems without considering the integration of social networks, where they assumed there are enough participants to perform crowdsourcing tasks.

Moreover, there are a few research works beginning to integrate MCS with social networks. For example, in \cite{jiang2017context} \cite{wang2018strategic}, they studied crowdsourcing task assignment problems in mobile social networks, where they exploited the neighborhood relationship in social networks to solve complex tasks. Recently, the method of task diffusion in social networks was utilized to remedy the dilemma of insufficient participation in MCS. Wang \textit{et al.} \cite{wang2018social} first proposed a social-network-assisted incentive mechanism to solve the worker recruitment problem by selecting a worker seed set to maximize the coverage, where they extended IC model and LT model for considering more specific factors. Wang \textit{et al.} \cite{wang2020socialrecruiter} designed a dynamic incentive mechanism to encourage limited workers to diffuse tasks in social networks, where they proposed a task-specific SIR epidemic model to simulate the diffusion process for crowdsourcing tasks. Xu \textit{et al.} \cite{xu2021incentive} studied to diffuse and complete crowdsourcing tasks while minimizing the total cost, where they presented a global influence estimation method to measure the influence spread in social networks. Wang \textit{et al.} \cite{wang2021acceptance} proposed a new acceptance-aware worker recruitment game in social-aware MCS based on a random diffusion model, where their goal is to maximize overall task acceptance through a meta-heuristic evolutionary approach. 

The above-mentioned four works of MCS in social networks are the most relevant to ours, however, it is different from ours for the following reasons: (1) For the diffusion model, they directly extend existing diffusion models, such as IC, LT, and SIR, which can only be targeted at a single information propagation cascade. Instead, we consider that multiple crowdsourcing tasks are diffused in the network at the same time. This is more practical in MCS. (2) For the incentive mechanism, they only focused on maximizing the coverage or utility, thus there is no guarantee of truthfulness. Here, we have both. (3) For the algorithm design, most of them use heuristic strategies to optimize the objective function, thus there is no theoretical bound. In contrast, our proposed algorithm has a constant approximation ratio and detailed theoretical analysis. Thus, they make our research problem and methods more practical and rigorous.

%Up to now, the existing auction-based incentive mechanisms can be classified into two parts. One is to maximize the social welfare in an MCS system. 

\section{System Models \& Problem Formulation}
An MCS system is usually composed of task requesters, workers, and the MCS platform, where our MCS is operated in a social network. Thus, we first introduce the preliminary knowledge of social networks. Then, we need to select a subset of registered users as workers to diffuse task information across the social network in order to attract more users for the crowdsourcing tasks. Thus, we introduce a reverse auction model to achieve the selection of workers. Finally, combined with the previous two parts, we formally define our Multi-Task Diffusion Maximization (MT-DM) problem and give its basic properties of hardness.
\subsection{Preliminaries of Social Networks}
A social network is represented by a directed graph $G=(V,E)$ with a node set $V=\{v_1,v_2,\cdots,v_n\}$ and an edge set $E=\{e_1,e_2,\cdots,e_m\}$, where each node $v\in V$ is an user and each edge $e=(u,v)\in E$ is a link from $u$ to $v$ representing some kind of relationship. For each edge $(u,v)\in E$, we say $u$ is an in-neighbor of $v$ and $v$ is an out-neighbor of $u$. For each node $v\in V$, we denoted by $N^-(v)$ the set of its in-neighbors and $N^+(v)$ the set of its out-neighbors.

In social networks, we can spread information by selecting a small number of seed users (or called seeds), which is the classical Influence Maximization (IM) problem \cite{kempe2003maximizing}. For the IM problem, the information diffusion from seeds is based on a predefined diffusion model, such as IC model and LT model. We say  a user is active if she accepts (is activated by) the information cascade from her in-neighbors or she is selected as a seed; otherwise, it is inactive. Here, we use the IC model as an example to demonstrate our problem, which can be easily extended to other diffusion models.
\begin{defn}[IC model]\label{defn1}
	It can be defined by a parameter set $W=\{w_{uv}:(u,v)\in E\}$, where there is a weight $w_{uv}$ associated with each edge $(u,v)\in E$. Given a seed set $S\subseteq V$, the process of information diffusion can be shown as follows: (1) At timestep $0$, all seed nodes are activated; (2) At timestep $t$, each node $u$ that is newly activated at timestep $t-1$ has a chance to activate its inactive out-neighbor $v$ with probability $w_{uv}$; and (3) The information diffusion terminates when no more inactive nodes can be activated.
\end{defn}

Given a social graph $G=(V,E)$ and IC model $W$, the IM problem is to find a seed set $S\subseteq V$ with at most $k$-size such that maximizing the influence spread $\sigma_W(S)$, which is the expected number of active nodes after information diffusion based on $W$. To clearly quantify the influence spread, we first give a concept called ``realization''. A realization $g=(V,E_g)$ is a subgraph with $E_g\subseteq E$ sampled on a given diffusion model. For the IC model $W$, each edge $(u,v)\in E$ will be placed in $E_g$ with probability $w_{uv}$, and those edges in $E_g$ are called live edges. Thus, the proability of realization $g$ sampled on $W$ is $\Pr[g;W]=\prod_{e\in E_g}w_e\prod_{e\in E\backslash E_g}(1-w_e)$. Like this, the influence spread can be considered as the expected number of activated nodes on all possible realizations. Now, the IM problem can be written in a formal expectation form: selecting a seed set $S^\circ$ such that
	\begin{align}
	S^\circ&\in\arg\max_{|S|\leq K}\sigma_W(S)\\
	&=\mathbb{E}_{g\sim W}[|I_g(S)|]=\sum_{g}\Pr[g;W]\cdot |I_g(S)|,
\end{align}
where $g\sim W$ implies that $g$ is sampled based on model $W$ and $I_g(S)$ is the set that contains all nodes can be reached from a node in $S$ through live edges in the realization $g$.

 Given a set function $f : 2^V\rightarrow\mathbb{R}_+$ and any two subsets $S$ and $T$ with $S\subseteq T\subseteq V$, we say it is monotone if $f(S)\leq f(T)$ and submodular if $f(S\cup\{v\})-f(S)\geq f(T\cup\{v\})-f(T)$ for any $v\in V\setminus T$. Although the IM problem is NP-hard, the influence spread $\sigma_W(S)$ is monotone and submodular w.r.t. $S$ under the IC model, thus the greedy hill-climbing algorithm can approach the optimal solution with $(1-1/e)$ approximation \cite{nemhauser1978analysis}. However, it is \#P-hard to compute the influence spread $\sigma_W(S)$ given a seed set $S$ \cite{chen2010scalable}, and then the hill-climbing algorithm can only return a $(1-1/e-\varepsilon)$ approximate solution within $\Omega(kmn\cdot poly(1/\varepsilon))$ running time by using Monte Carlo simulations.

\subsection{Reverse Auction Model}
The MCS platform publishes a task set, denoted by $T=\{t_1,t_2,\cdots,t_{n_T}\}$, which contains $n_T$ different tasks. For each task $t_j\in T$, it can be a specific type of location-dependent task such as data collection, traffic monitoring, distributed computing, and so on. In our setting, the task type will directly affect the spread of this task in social networks by affecting the model weights in the IC model. Thus, the spread of each task in social networks is different.

Without loss of generality, we assume there is a set of registered users $V_R\subseteq V$, denoted by $V_R=\{v_1,v_2,\cdots,v_{n_R}\}$, who are interested in undertaking a group of crowdsourcing tasks and diffusing them in social networks. For each registered user $v_i\in V_R$, it gives its bidding information $B_i=\{T_i,b_i\}$ to the MCS platform, where $T_i\subseteq T$ is the task set it claims and $b_i$ is the bid price to claim $T_i$. For the task set $T_i$, there is a cost $c_i$ that indicates the real cost for user $v_i$ to undertake the tasks in $T_i$, which is the private information and maybe not equal to the claimed price $b_i$. After receiving the bidding information of all registered users $\{B_i\}_{v_i\in V_R}$, the MCS platform will select a subset of registered users, denoted by $S\subseteq V_R$, as winners. Each winner $v_i$ will execute its claimed tasks in $T_i$ and diffuse these tasks in social networks to attract more users to participate in crowdsourcing.

Given a social network $G=(V,E)$, a task set $T$, a diffusion model $W'$, a registered user set $V_R\subseteq V$, and the bidding information from all registered users $\vec{B}=\{B_i\}_{v_i\in V_R}$, we need to design an incentive mechanism $\mathcal{M}(G, T, W', V_R,\vec{B})$ that return a winning set $S\subseteq V_R$ and a payment set $\vec{P}=\{p_i\}_{v_i\in S}$ for each registered user in $S$. For each winning registered user $v_i\in S$, we define its utility as the difference between the payment and cost, that is
\begin{equation}
	u_i=p_i-c_i.
\end{equation}
For each losing registered user $v_{i'}\in V_R\backslash S$, we define its utility $u_{i'}=0$ since there is no payment and no cost. Due to the fact that registered users in the social network are selfish and rational, they tend to maximize their utilities by manipulating their bids (giving dishonest bids). Besides, we assume that the task sets submitted by all registered users can cover all tasks, that is $T=\cup_{v_i\in V_R}T_i$. This assumption is reasonable for the MCS platform since we can remove the task from $T$ if there is no user bidding on it. 

\subsection{Problem Definition}
To quantify the completion quality of a crowdsourcing task, we divide the area $A_G$ covered by the social network $G$ into $n_A$ sub-areas, denoted by $A_G=\{a_1,a_2,\cdots, a_{n_A}\}$. For each user $v_i\in V$, we can get its location information, denoted by $(x_i,y_i)$, according to its IP address. From this coordinate, we can judge which sub-area the user belongs to. Thus, we define a location mapping $\mathcal{L}: V\rightarrow A_G$ that returns the sub-area a user belongs to. For example, $\mathcal{L}(v)=a_k$ indicates the user $v$ is in sub-area $a_k$. For each task $t_j\in T$, its completion quality in different sub-area may be the same or different. Thus, we define a quality set $\vec{q}_j=\{q_j^1,q_j^2,\cdots,q_j^{n_A}\}$, where each element $q_j^k\in \vec{q}_j$ represents the expected completion quality of the task $t_j$ in sub-area $a_k$. The completion quality set of all tasks is denoted by $\vec{Q}_T=\{\vec{q}_j\}_{t_j\in T}$.

Based on $\mathcal{M}(G, T, W', V_R,\vec{B})$, we should consider how multiple tasks diffuse, namely to define $W'$. Here, we regard that each task can be independently spread in the same social network, and we name it as Multi-Task IC model.
\begin{defn}[Multi-Task IC model]\label{defn2}
	Given a social graph $G=(V,E)$ and a task set $T$, it has a equivalent multi-layer graph $G'=(V',E')$ with $G'=G^1\cup G^2\cup\cdots\cup G^{n_T}$, where we make a copy $G^j=(V^j,E^j)$ of $G$ with $V^j=V$ and $E^j=E$ for each task $t_j\in T$. For each task diffusion layer $G^j$, it has a parameter set $W^j=\{w_{uv}^j:(u,v)\in E\}$ and task $t_j$ is spread on $G^j$ based on the IC model defined in Definition \ref{defn1}. The diffusion process of each layer is independent to other layers. For convenience, we can denote a Multi-Task IC model by parameter sets $\vec{W}_T=\{W^j\}_{t_j\in T}$.
\end{defn}
Based on the Multi-Task IC model, given a winning seed set $S\subseteq V_R$, we can project it to each task diffusion layer $G^j$. Thus, we define a projection $\mathcal{P}$ as
\begin{equation}\label{eq4}
	\mathcal{P}(S; G^j)=S^j=\{v_i\in S: t_j\in T_i\},
\end{equation}
where $S^j\subseteq S$ is the user set that claims the task $t_j$ in its bidding information. Now, we can formally define the Multi-Task Diffusion Function as follows.
\begin{defn}[Multi-Task Diffusion Function]\label{def3}
	Given a social graph $G=(V,E)$, a task set $T$, its corresponding Multi-Task IC model $\vec{W}_T$, and completion quality set of all tasks $\vec{Q}_T$, the tasks diffusion function $f_{\vec{W}_T}(S)$ of a winning seed set $S$ can be defined as
	\begin{align}
		&f_{\vec{W}_T}(S)=\frac{1}{n_T}\sum_{t_j\in T}\sum_{v\in V}\Pr\left[v\in I_{W^j}(S^j)\right]\cdot q_j^{\mathcal{L}(v)}\\
		&=\frac{1}{n_T}\sum_{t_j\in T}\sum_g\sum_{v\in I_g(S^j)}\Pr[g; W^j]\cdot q_j^{\mathcal{L}(v)}\\
		&=\frac{1}{n_T}\sum_{t_j\in T}\sum_g\sum_{v\in V}\Pr[g; W^j]\cdot\mathbb{I}\left[v\in I_g(S^j)\right]\cdot q_j^{\mathcal{L}(v)},\label{eq6}
	\end{align}
	where $I_{W^j}(S^j)$ is a random variable representing the set of activated nodes by $S^j$ on task diffusion layer $G^j$ with model parameter $W_j$, and $\mathbb{I}[x]$ is an indicator with  $\mathbb{I}[x]=1$ if $x$ is true; otherwise $\mathbb{I}[x]=0$.
\end{defn}
\noindent
Here, each task is independently spread on its task diffusion layer, and the Multi-Task Diffusion Function can be considered as the expected completion quality of all tasks after finishing diffusion in the social network. we are enough to define the Multi-Task Diffusion Maximization (MT-DM) problem.
\begin{pro}[MT-DM]
	Given a social network $G=(V,E)$, a task set $T$, its corresponding Multi-Task IC model $\vec{W}_T$, a registered user set $V_R\subseteq V$, and the bidding information $\vec{B}$, the MT-DM problem is to design an incentive mechanism $\mathcal{M}(G, T,\vec{W}_T, V_R,\vec{B}, D)$ that selects a winning seed set $S^\circ\subseteq V_R$ under the budget $D$ to maximize the Multi-Task Diffusion Function. It can be formalized as
	\begin{equation}
		S^\circ\in\arg\max_{S\subseteq V_R}f_{\vec{W}_T}(S) \text{ s.t. }\sum_{v_i\in S}p_i\leq D,
	\end{equation}
	where $\vec{P}=\{p_i\}_{v_i\in V_R}$ is the payment returned by $\mathcal{M}$.
\end{pro}

When the context is clear, we can abbreviate $f_{\vec{W}_T}(S)$ as $f(S)$. The MT-DM problem remains at least the same hardness as solving the IM problem \cite{kempe2003maximizing}, which will be shown in the following two theorems.
\begin{thm}\label{thm1}
	The MT-DM problem is NP-hard and given a seed set $S$, computing the objective function $f(S)$ defined in Eqn. (\ref{eq6}) is \#P-hard.
\end{thm}
\begin{proof}
	Consider the simplest case, by setting $|T|=1$ and the expected completion quality of this task in all sub-areas is $1$, that is $q^k=1$ for each subarea $a_k\in A_G$, the MT-DM problem can be reduced to the IM problem. Thus, it inherits the NP-hardness of the IM problem \cite{kempe2003maximizing} and \#P-hardness of computing the influence spread under the IC model \cite{chen2010scalable}.
\end{proof}
\begin{thm}\label{thm2}
	The objective function $f(S)$ defined in Eqn. (\ref{eq6}) is monotone and submodular with respect to $S$.
\end{thm}
\begin{proof}
	The monotonicity is obvious, and we only need to show the submodularity. Considering any two subset $S_1\subseteq S_2\subseteq V_R$ and any user $u\in V_R\backslash S_2$, based on Eqn. (\ref{eq6}), this is equivalent to show $\mathbb{I}[v\in I_g(S_1^j\cup\{u\})]-\mathbb{I}[v\in I_g(S_1^j)]\geq\mathbb{I}[v\in I_g(S_2^j\cup\{u\})]-\mathbb{I}[v\in I_g(S_2^j)]$. When $\mathbb{I}[v\in I_g(S_2^j\cup\{u\})]-\mathbb{I}[v\in I_g(S_2^j)]=1$, we have $\mathbb{I}[v\in I_g(S_2^j\cup\{u\})]=1$ and $\mathbb{I}[v\in I_g(S_2^j)]=0$, which implies that there exists a path from $u$ to $v$ in the realization $g$. Thus, we have $\mathbb{I}[v\in I_g(S_1^j)]=0$ since $S^j_1\subseteq S^j_2$  based on Eqn. (\ref{eq4}) and $S_1\subseteq S_2$, and we have $\mathbb{I}[v\in I_g(S_1^j\cup\{u\})]=1$ since there exists a path from $u$ to $v$ in the realization $g$. Then, the objective function $f(S)$ can be decomposed into the linear combination of $\mathbb{I}[v\in I_g(S^j)]$, therefore it is monotone and submodular.
\end{proof}

\section{Incentive Mechanism For MT-DM}
For the auction model of our MT-DM problem, the MCS platform that holds the task set $T$ is the buyer, and the registered users in $V_R$ are the seller. In this reverse auction, the incentive mechanism between the MCS platform and users should satisfy several design rationales. Thus, in this section, we first give four design rationales as follows, and then describe our mechanism design and show how it satisfies these four design rationales.
\begin{itemize}
	\item Individual Rationality: For each user, the payment it gets should be larger than its cost. That is, for each $v_i\in S$, we have $u_i=p_i-c_i\geq 0$.
	\item Budget Balance: In our TDM problem, it indicates that the total payment paid to the winning workers will not be larger than the predefined budget $D$.
	\item Truthfulness: In order to ensure a fair trading environment, truthfulness is the most important principle to design an auction algorithm. It guarantees that no user can get more benefit by giving an untruthful bid, thus bidding with the real cost is the dominant strategy for each user.
	\item Computational efficiency: The auction algorithm can be executed in polynomial time.
\end{itemize}

Here, the task set $T_i$ submitted by user $v_i$ must be truthful because every user must perform its claimed task once winning the auction. Thus, we consider a piece of bidding information $B_i=\{T_i,b_i\}$ as a truthful bid if $b_i=c_i$, otherwise consider it as an untruthful bid if $b_i\neq c_i$.

\subsection{Mechanism Design}
To maximize the Multi-Task Diffusion Function, it is best to assign the payment $p_i$ to each user by its bid $b_i$ because this is beneficial to select as many users as possible to undertake crowdsourcing tasks, and at the same time meet the requirement of individual rationality. However, if doing this, we cannot satisfy the truthfulness. Thus, we select to relax the budget balance, where we change the constraint $\sum_{v_i\in S}p_i\leq D$ to $\sum_{v_i\in S}b_i\leq D$. Even though the budget balance will not be strictly satisfied, we will find later when the competition is sufficient, namely $|V_R|\gg|S|$, a slight increase in the bid will cause the user to lose the auction. Thus, for each winner $v_i\in S$, $p_i$ will be very close to $b_i$.

\begin{algorithm}[!t] 
	\caption{{MT-DM-L $(G, T,\vec{W}_T, V_R,\vec{B}, D)$}}
	\label{a1}
	\KwIn{Social graph $G=(V,E)$, task set $T$, Multi-Task IC model $\vec{W}_T$, registered user set $V_R$, bidding information $\vec{B}$, and budget $D$.}
	\KwOut{Winning set $S$ and their payments $\vec{P}$.}
	\tcp{Winning worker selection stage}
	Initialize: $S_0\leftarrow\emptyset$, $sum\leftarrow 0$, $i\leftarrow0$\;
	\While{$S_i\neq V_R$}{
		Select $v_{i+1}\in\arg\max_{v_k\in V_R\backslash S_i}f(v_k|S_i)/b_k$\;
		\If{$sum+b_{i+1}>D$}{
			Break\;
		}
		$S_{i+1}\leftarrow S_i\cup\{v_{i+1}\}$\;
		$sum\leftarrow sum+b_{i+1}$\;
		$i\leftarrow i+1$\;
	}
	$S\leftarrow S_i$\;
	\tcp{Payment determination stage}
	\ForEach{$v_i\in S$}{
		Initialize: $H_0\leftarrow\emptyset$, $sum\leftarrow 0$, $l\leftarrow 0$\;
		Initialize: $p_i\leftarrow-\infty$\;
		\While{$H_l\neq V_{R:-i}$}{
			Select $v_{i_{l+1}}\in\arg\max_{v_{i_k}\in V_{R:-i}\backslash H_l}f(v_{i_k}|H_l)/b_{i_k}$\;
			$H_{l+1}\leftarrow H_l\cup\{v_{i_{l+1}}\}$\;
			$sum\leftarrow sum+b_{i_{l+1}}$\;
			$p_i\leftarrow\max\left\{p_i,b_{i_{l+1}}\cdot\frac{f(v_i|H_l)}{f(v_{i_{l+1}}|H_l)}\right\}$\;
			$l\leftarrow l+1$\;
			\If{$sum+b_i>D$}{
				Break\;
			}
		}
		\If{$sum+b_i\leq D$}{
			$p_i\leftarrow\max\{p_i, D-sum\}$\;
		}
	}
	\Return $S$ and $\vec{P}=\{p_i\}_{v_i\in S}$
\end{algorithm}

The running process of our mechanism design is shown in Algorithm \ref{a1}, named as MT-DM-L. It has two stages: winning worker selection stage and payment determination stage. In the winning worker stage, shown as line 1 to line 9 of Algorithm \ref{a1}, we adopt the greedy strategy to select a subset of registered users as winning workers. For each user $v_k\in V_R$, its marginal gain from unit bid based on current winning set $S_i$ can be defined $f(v_k|S_i)/b_k=(f(S_i\cup\{v_k\})-f(S_i))/b_k$. In each iteration, it selects the registered user with maximum marginal gain from unit bid based on the candidate set $V_R\backslash S_{i}$ until using up all the budget. Without loss of generality, the registered user set can be denoted by a list $V_R=\{v_1,v_2,\cdots,v_{n_R}\}$ sorted by the order of selection in a greedy manner. Thus, we have $f(v_1|S_0)/b_1\geq f(v_2|S_1)/b_2\geq\cdots\geq f(v_{n_R}|S_{n_R-1})/b_{n_R}$, where $v_i\in\arg\max_{v_k\in V_R\backslash S_{i-1}}f(v_k|S_{i-1})$ and $S_{i-1}=\{v_1,v_2,\cdots,v_{i-1}\}$. The winning worker set $S$ is returned by the winning worker selection stage in line 9 of Algorithm \ref{a1}. Assuming there are $n_S$ users in $S$, it can be denoted by $S=\{v_1,v_2,\cdots,v_{n_S}\}$ based on the order of $V_R$.

In the payment determination stage, it is shown as line 10 to line 22 of Algorithm \ref{a1}. For each winning worker $v_i\in S$, we greedily select a user set $H$ from $V_{R:-i}=V_R\backslash\{v_i\}$. Without loss of generality, the $V_{R:-i}$ can be denoted by a list $V_{R:-i}=\{v_{i_1},v_{i_2},\cdots,v_{i_{n_R-1}}\}$ sorted by the order of selection in a greedy manner. Thus, we have $f(v_{i_1}|H_0)/b_{i_1}\geq f(v_{i_2}|H_1)/b_{i_2}\geq\cdots\geq f(v_{n_R-1}|H_{n_R-2})/b_{n_R-1}$, where $v_{i_l}\in\arg\max_{v_{i_k}\in V_{R:-i}\backslash H_{l-1}}f(v_{i_k}|H_{l-1})$ and $H_{l-1}=\{v_{i_1},v_{i_2},\cdots,v_{i_{l-1}}\}$. Let $H=\{v_{i_1},v_{i_2},\cdots,v_{i_{n_H}}\}$ be the set generated in the ``while'' loog shown as line 13 to line 20 of Algorithm \ref{a1}. We denote by $n_H$ the size of $H$, and it is sufficient to consider the following two cases:
\begin{itemize}
	\item Case 1: $\sum_{k=1}^{n_H}b_{i_k}+b_i>D$. It indicates that $n_H\leq n_R-1$ and $\sum_{k=1}^{n_H-1}b_{i_k}+b_i\leq D$. If we use $v_i$ to replace a user $v_{i_a}\in H$, the total bid of $\{v_{i_1},v_{i_2},\cdots,v_{i_{a-1}},v_i\}$ will not exceed the budget $D$ since we have $\sum_{k=1}^{a-1}b_{i_k}+b_i\leq\sum_{k=1}^{n_H-1}b_{i_k}+b_i<D$. We define $b_{i(a)}$ as the maximum bid of user $v_i$ to replace a user $v_{i_a}$ in $H$, which must hold $f(v_i|H_{a-1})/b_{i(a)}\geq f(v_{i_a}|H_{a-1})/b_{i_a}$. Thus, we have the critical bid of user $v_i$ as
	\begin{equation}\label{eq8}
		b_{i(a)}=b_{i_a}\cdot f(v_i|H_{a-1})/f(v_{i_a}|H_{a-1}).
	\end{equation}
	At this time, if the user $v_i$ cannot replace any user in $H$, it will lose the auction because we have the budget constraint $\sum_{k=1}^{n_H}b_{i_k}+b_i>D$. Thus, the payment $p_i$ can be defined as
	\begin{equation}\label{eq9}
		p_i=\max\left\{b_{i(1)},b_{i(2)},\cdots,b_{i(n_H)}\right\}.
	\end{equation}
	
	\item Case 2: $\sum_{k=1}^{n_H}b_{i_k}+b_i\leq D$. It indicates that $n_H=n_R-1$. If the user $v_i$ can replace a user in $H$, it will win the auction definitely. If the user $v_i$ cannot replace any user in $H$, it will win as well because there is enough budget. Thus, the payment $p_i$ can be defined as
	\begin{equation}\label{eq10}
		p_i=\max\left\{\max_{1\leq k\leq n_H}\left\{b_{i(k)}\right\},D-\sum\nolimits_{k=1}^{n_H}b_{i_k}\right\}.
	\end{equation}
\end{itemize}

\subsection{Mechanism Analysis}
Next, we show our MT-DM-L holds the desired properties of individual rationality, truthfulness, and computational efficiency. Besides, the MT-DM problem can be solved with an effective approximation.
\begin{lem}\label{lem1}
	The MT-DM-L is individually rational.
\end{lem}
\begin{proof}
	For each winning user $v_i\in S$, it is the $i$-th user in sorted list $V_R$. For the first case, it must have $i\leq n_H$, otherwise it will lose the auction. Thus, it exist a user $v_{i_a}\in H$ such that $f(v_i|H_{a-1})/b_i\geq f(v_{i_a}|H_{a-1})/b_{i_a}$. Based on the pricing defined in Eqn. (\ref{eq8}) and Eqn. (\ref{eq9}), we have $p_i\geq b_{i_a}\cdot f(v_i|H_{a-1})/f(v_{i_a}|H_{a-1})\geq b_i$. Thus, it has $u_i=p_i-c_i>0$ if the user $v_i$ bids truthfully. For the second case, we have $p_i\geq D-\sum_{k=1}^{n_H}b_{i_k}\geq b_i$. Therefore, the utility is always positive and individual rationality holds.
\end{proof}

Based on Myerson theory \cite{nisan2007algorithmic}, a reverse auction is truthful if and only if
\begin{itemize}
	\item The bid is monotone. For each user $v_i\in V_R$, if it wins the auction by bidding $B_i=(T_i,b_i)$, it will win by any bidding $B'_i=(T_i,b'_i)$ with $b'_i<b_i$.
	\item The payment is critical. For each winning worker $v_i\in S$, its payment $p_i$ is the maximum bid it can win.
\end{itemize}
\begin{lem}\label{lem2}
	The MT-DM-L is truthful.
\end{lem}
\begin{proof}
	Let us first verify that the bid is monotone. For each user $v_i\in V_R$ with bidding $B_i=(T_i,b_i)$, its position in the list $V_R$ sorted by the order of selection in a greedy manner will be moved towards the head when it gives a bid $b'_i$ with $b'_i<b_i$. Secondly, the size of the winning set $S'$ will not shrink when the user $v_i$ gives a bid $b'_i$ with $b'_i<b_i$. Therefore, if the user $v_i$ wins the auction by bidding $b_i$, it will win as well by bidding $b'_i$ with $b'_i<b_i$.
	
	Next, we verify the payment to each winning worker is critical. For each winning worker $v_i\in S$, it is the $i$-th user in sorted list $V_R$. For the first case, it must have $i\leq n_H$, otherwise it will lose the auction. If its bid $b'_i>p_i=\max_{1\leq k\leq n_H}\{b_{i(k)}\}$, we have $f(v_i|H_{a-1})/b'_i<f(v_{i_a}|H_{a-1})/b_{i_a}$ for any $v_{i_a}\in H$.	Thus, the user $v_i$ cannot replace any user in $H$ and $i>n_H$ when giving a bid $b'_i$. Because of the budget constraint $\sum_{k=1}^{n_H}b_{i_k}+b'_i>D$, the user $v_i$ will lose the auction. For the second case,  if its bid $b'_i>p_i$, we have $b'_i\geq\max_{1\leq k\leq n_H}\{b_{i(k)}\}$ where the user $v_i$ cannot replace any user in $H$. Since $n_H=n_R-1$, the user $v_i$ is the $n_R$-th user in sorted list $V_R$. Here, we also have $b'_i>D-\sum_{k=1}^{n_H}b_{i_k}$. The user $v_i$ will lose the auction if there is not enough remaining budget.
\end{proof}

\begin{lem}\label{lem3}
	The MT-DM-L is computationally efficient.
\end{lem}
\begin{proof}
	According to the budget constraint in Algorithm \ref{a1}, the number of winning workers is at most $D/\min_{v_i\in V_R}\{b_i\}$. In the winning worker selection stage, in each iteration, it takes $\mathcal{O}(n_R)$ time to select the user with maximum marginal gain from unit bid. Thus, the winning worker selection stage takes overall $(D\cdot n_R/\min_{v_i\in V_R}\{b_i\})$. In the payment determination stage, for each winning worker $v_i\in S$, it takes $\mathcal{O}(D\cdot n_R/\min_{v_i\in V_R}\{b_i\})$ to determine its price. Thus, the payment determination stage takes overall $\mathcal{O}(D\cdot n_R)$ time. Therefore, the total running time can be bounded by $\mathcal{O}(D\cdot n_R)$, and then the MT-DM-L is computationally efficient.
\end{proof}

For the convenience of analysis, we give an assumption that the bid of each registered user is far less than the budget, namely $\max_{v_i\in V_R}\{b_i\}\ll D$. In fact, this assumption is reasonable. If the budget is only enough to pay very few users, this problem will become meaningless. Based on this assumption, we have the following lemma.
\begin{lem}
	The MT-DM-L can return a winning set within $(1-1/\sqrt{e})$ approximation.
\end{lem}
\begin{proof}
	According to Theorem \ref{thm1} and Theorem \ref{thm2}, the MT-DM problem is NP-hard and the tasks diffusion function is monotone and submodular. Based on the Greedy algorithm proposed in \cite{khuller1999budgeted}, it considers two candidate solutions and selects the best one. The first solution can be obtained by sequentially selecting the users with the maximum ratio between their marginal gain and unit cost (bid in this paper) until there is not enough remaining budget to select the next user. The second solution can be obtained by selecting the user with the maximum benefit within the budget. According to our assumption $\max_{v_i\in V_R}\{b_i\}\ll D$, it is impossible for an extreme case that the benefit of a single user is better than the total benefit of a sequence of users. Thus, the second solution will not be selected and our winning worker selection stage can return a solution with $(1-1/\sqrt{e})$ approximation.
\end{proof}

\section{Diffusion Function Estimation}
Even though the winning worker selection stage of MT-DM-L can return a solution with $(1-1/\sqrt{e})$ approximation, which has been proven in the last section, it is \#P-hard to compute the Multi-Task Diffusion Function $f(S)$ given a worker set $S$ based on Theorem \ref{thm1}. This leads to high computational cost by implementing it by Monte Carlo simulations. In order to reduce the time complexity without losing too much approximation, we will use the technique of reverse influence sampling (RIS) in our problem.

\subsection{Multi-Task Sampling}
The RIS is achieved by using random reverse reachable (RR) sets. Back to our MT-DM problem, it relies on the Multi-Task IC model $\vec{W}_T$ where the original graph $G$ has been expressed by multi-layer graph $G'$ with $G'=G^1\cup G^2\cup\cdots\cup G^{n_T}$. For each layer $G^j$, there is a corresponding IC model $W^j$ on it. Given the model $W^j$, a random RR set $R^j$ based on the Benefit Sampling Algorithm (BSA) \cite{nguyen2017billion} can be generated in the following three steps:
\begin{itemize}
	\item Select a node $u$ with probability $q_j^{\mathcal{L}(u)}/\sum_{v\in V}q_j^{\mathcal{L}(v)}$.
	\item Sample a realization $g$ based on IC model $W^j$.
	\item Collect the nodes that can reach $u$ through a live path in the realization $g$.
\end{itemize}
For each node $v\in V$, the probability that it is contained in $R^j$ generated from node $u$ on $W^j$ equals the probability that $v$ can activate $u$. For convenience, we define
\begin{equation}
	f^j(S^j)=\sum_g\sum_{v\in I_g(S^j)}\Pr[g; W^j]\cdot q_j^{\mathcal{L}(v)}.
\end{equation}
Thus, we have $f(S)=\frac{1}{n_T}\sum_{t_j\in T}f^j(S^j)$. We can reformulate the following lemma based on \cite{nguyen2017billion}.
\begin{lem}[\cite{nguyen2017billion}]\label{lem5}
	Let $S^j$ be the projection of $S$ on $G^j$ and $R^j$ be a random RR set generated on $W^j$, then we have
	\begin{equation}
		f^j(S^j)=\left(\sum\nolimits_{v\in V}q_j^{\mathcal{L}(v)}\right)\cdot\mathbb{E}_{R^j}\left[\mathbb{I}(S^j\cap R^j)\right],
	\end{equation}
	where we have $\mathbb{I}(S^j\cap R^j)=1$ if $S^j\cap R^j\neq\emptyset$, otherwise $\mathbb{I}(S^j\cap R^j)=0$.
\end{lem}

Now, we can formally define the concept of random multi-task reverse reachable (MT-RR) set. Given a social graph $G$ and Multi-Task IC model $\vec{W}_T$, a random MT-RR set can be generated in the following two steps:
\begin{itemize}
	\item Uniformly select a task $t_x\in T$.
	\item Generate a random RR set on $W^x$ by following the above procedure.
\end{itemize}
Thus, a random MT-RR set can be denoted by $R(X)$, where $X$ is a random variable that implies it is generated from task $t_X$. Then, we can use a collection of random MT-RR sets to efficiently estimate our Multi-Task Diffusion Function, which is shown in the following theorem.
\begin{thm}\label{thm3}
	Let $S\subseteq V_R$ be a seed set and $R(X)$ be a random MT-RR set, then we have
	\begin{equation}\label{eq14}
		f(S)=\sum_{t_j\in T}\sum_{v\in V}q_j^{\mathcal{L}(v)}\cdot\mathbb{E}_{R(X)}\left[\mathbb{I}_j(S^j\cap R(X))\right],
	\end{equation}
	where we have $\mathbb{I}_j[S^j\cap R(X)]=1$ if $t_j=t_{X}$ and $S^j\cap R(X)\neq\emptyset$, otherwise $\mathbb{I}_j[S^j\cap R(X)]=0$.
\end{thm}
\begin{proof}
	Considering a fixed task $t_j\in T$, we first look at the property of $\mathbb{E}_{R(X)}[\mathbb{I}_j(S^j\cap R(X))]$. Due to the randomness of $X$, we have $\mathbb{E}_{R(X)}[\mathbb{I}_j(S^j\cap R(X))]=$
	\begin{align}
		&=\Pr[X=j]\cdot\mathbb{E}_{R(X)}[\mathbb{I}_j(S^j\cap R(X))|X=j]\nonumber\\
		&\quad+\Pr[X\neq j]\cdot \mathbb{E}_{R(X)}[\mathbb{I}_j(S^j\cap R(X))|X\neq j]\label{eq15}\\
		&=\Pr[X=j]\cdot\mathbb{E}_{R^j}[\mathbb{I}_j(S^j\cap R^j)]\label{eq16}\\
		&=(1/n_T)\cdot\mathbb{E}_{R^j}[\mathbb{I}(S^j\cap R^j)]\label{eq17},
	\end{align}
	where we have $\mathbb{E}_{R(X)}[\mathbb{I}_j(S^j\cap R(X))|X\neq j]=0$ in Eqn. (\ref{eq16}) and $\Pr[X=j]=1/n_T$ in Eqn. (\ref{eq17}). Then, substituting the $\mathbb{E}_{R(X)}[\mathbb{I}_j(S^j\cap R(X))]$ in $f(S)$, we have
	\begin{align}
		f(S)&=\frac{1}{n_T}\sum_{t_j\in T}f^j(S^j)\\
		&=\frac{1}{n_T}\sum_{t_j\in T}\left(\sum\nolimits_{v\in V}q_j^{\mathcal{L}(v)}\right)\cdot\mathbb{E}_{R^j}\left[\mathbb{I}(S^j\cap R^j)\right]\label{eq19}\\
		&=\sum_{t_j\in T}\sum\nolimits_{v\in V}q_j^{\mathcal{L}(v)}\cdot\mathbb{E}_{R(X)}\left[\mathbb{I}_j(S^j\cap R(X))\right],
	\end{align}
	where Eqn. (\ref{eq19}) is based on Lemma \ref{lem5}. Then, this theorem has been proven.
\end{proof}

Based on Theorem \ref{thm3}, we can generate a collection of random MT-RR sets, denoted by $\mathcal{R}=\{R^{j_1}_1,R^{j_2}_2,\cdots,R^{j_\theta}_\theta\}$. Given a collection $\mathcal{R}$, we can define an estimation function to efficiently estimate the tasks diffusion function $f(S)$, which is denoted by $\hat{f}(S;\mathcal{R})$. Then, we have
	\begin{align}
		\hat{f}(S;\mathcal{R})&=\sum_{t_j\in T}\left(\sum\nolimits_{v\in V}q_j^{\mathcal{L}(v)}\right)\cdot\frac{1}{\theta}\sum_{x=1}^{\theta}\mathbb{I}_j(S^j\cap R^{j_x}_x)\\
		&=\frac{1}{\theta}\sum_{x=1}^{\theta}\sum_{t_j\in T}\left(\sum\nolimits_{v\in V}q_j^{\mathcal{L}(v)}\right)\cdot\mathbb{I}_j(S^j\cap R^{j_x}_x)\\
		&=\frac{1}{\theta}\sum_{x=1}^{\theta}\left(\sum\nolimits_{v\in V}q_{j_x}^{\mathcal{L}(v)}\right)\cdot\mathbb{I}(S^{j_x}\cap R^{j_x}_x).\label{eq23}
	\end{align}
Here, the $\hat{f}(S;\mathcal{R})$ is an unbiased estimation of $f(S)$. It is a monotone and submodular function as well by giving a similar induction of Theorem \ref{thm2}.

\subsection{Algorithm Design}
Based on the technique of RIS, there are a lot of sampling-based algorithms such as TIM/TIM+ \cite{tang2014influence}, IMM \cite{tang2015influence}, SSA/D-SSA \cite{nguyen2016stop}, OPIM-C \cite{tang2018online}, and HIST \cite{guo2020influence}, which can solve the IM problem with $(1-1/e-\varepsilon)$ approximation and efficient polynomial time. According to the multi-task sampling and its corresponding unbiased estimation $\hat{f}(S;\mathcal{R})$ for the Multi-Task Diffusion Function $f(S)$, we can adopt a similar strategy to the winning worker selection stage of M-TDM. However, the RIS-based IM algorithms cannot be directed used to solve our problem because of the following two reasons: (1) The multi-task sampling is different from the RIS, thus the estimation process is different; and (2) Our budget is a knapsack constraint, not a cardinality constraint. Thus, we need to make appropriate adjustments according to the existing RIS-based algorithms to adapt to our problem. Here, we take the most recent OPIM-C algorithm as an example to illustrate our modification steps, and other algorithms can be modified by similar steps.

\begin{algorithm}[!t] 
	\caption{{Weighted-Max-Coverage $(\mathcal{R}, V_R, K)$}}
	\label{a2}
	Initialize: $S_0^*\leftarrow 0$\;
	\For{$a=1$ to $K$}{
		Select $v^*_a\in\arg\max_{v_k\in V_R\backslash S^*_{a-1}}\hat{f}(v_k|S^*_{a-1};\mathcal{R})$\;
		$S^*_a\leftarrow S^*_{a-1}\cup\{v^*_a\}$\;
	}
	\Return $S^*_K$\;
\end{algorithm}

\textbf{The first step} is to transform the MT-DM problem under the knapsack constraint to that under the cardinality constraint. Given the budget $D$, we can get an upper bound 
\begin{equation}\label{eq24}
	K=\max\left\{k:\exists S\subseteq V_R, |S|=k, \sum_{v_i\in S}b_i\leq D\right\},
\end{equation}
which is the maximum number of nodes under the budget $D$. Without loss of generality, we have $K\leq n_R/2$ based on the assumption of sufficient competition such that $|V_R|\gg|S|$. Thus, a feasible solution $S^*_D$ returned by any algorithms under the budget $D$ will not contain more than $K$ nodes, where the number of such seed sets with the same number of nodes can be bounded by $\binom{n_R}{|S^*_D|}\leq\binom{n_R}{K}$. 

Then, we transfer to optimize our objective function under the constant $K$ instead of the budget $D$. The number of random MT-RR sets that can provide an approximation guarantee under the constant $K$ is sufficient to provide the same approximation guarantee under the budget $D$ since the number of nodes contained in the seed set becomes smaller. This is the reason why we can do like this.

\textbf{The second step} is to calculate the weighted maximum coverage under the constant $K$ given a collection of random MT-RR set, which is shown as Algorithm \ref{a2}. Let $\hat{S}^*_K$ be the solution returned by the Weighted-Max-Coverage under the constant $K$, $\hat{S}^\circ_K$ be the optimal solution such that maximizing the estimation $\hat{f}(\cdot;\mathcal{R})$ under the constant $K$, and $S^\circ_K$ be the optimal solution of our objective $f(\cdot)$ under the constant $K$. Thus, we have
\begin{align}
	\hat{f}(\hat{S}^*_K;\mathcal{R})&\geq(1-1/e)\cdot\hat{f}(\hat{S}^\circ_K;\mathcal{R})\\
	&\geq(1-1/e)\cdot\hat{f}(S^\circ_K;\mathcal{R}),\label{eq26}
\end{align}
where it is because the $\hat{f}(S;\mathcal{R})$ is monotone and submodular with respect to $S$. Then, we have the following concentration bound according to the martingale analysis in \cite{tang2015influence}.
\begin{lem}[\cite{tang2015influence}]\label{lem6}
	For any $\xi>0$, given a seed set $S$ and a collection of random MT-RR sets $\mathcal{R}$, we have
	\begin{align}
		&\Pr\left[\hat{f}(S; \mathcal{R})\leq(1-\xi) f(S)\right]\leq\exp\left(-\frac{\xi^2\theta f(S)}{2}\right),\label{eq21}\\
		&\Pr\left[\hat{f}(S; \mathcal{R})\geq(1+\xi) f(S)\right]\leq\exp\left(-\frac{\xi^2\theta f(S)}{2+\frac{2}{3}\xi}\right).\label{eq22}
	\end{align}
\end{lem}

\textbf{The third step} is to achieve a collection of random MT-RR sets that contains a sufficient number of random MT-RR sets to estimate our objective function and guarantee the approximation by following the process of OPIM-C \cite{tang2018online}. Here, we need to modify it according to the estimation in the updated concentration bound in Lemma \ref{lem6}. First, we sample a collection of random MT-RR sets $\mathcal{R}_1$ to greedily select a size-$K$ seed set $\hat{S}^*_K$ based on Algorithm \ref{a2}, which can be used to obtain an upper bound $\overline{f}(S^\circ_K)$ of $f(S^\circ_K)$. Second, we sample another collection $\mathcal{R}_2$ with $|\mathcal{R}_2|=|\mathcal{R}_1|$, which can be used to obtain an lower bound $\underline{f}(\hat{S}^*_K)$ of $f(\hat{S}^*_K)$. If we have
\begin{equation}
	\underline{f}(\hat{S}^*_K)/\overline{f}(S^\circ_K)\geq(1-1/e-\varepsilon),
\end{equation}
then the algorithm will terminate and return the solution $\hat{S}^*_K$. Based on Lemma 4.2 in \cite{tang2018online}, the lower bound $\underline{f}(\hat{S}^*_K)$ on $\mathcal{R}_2$ can be modified as
\begin{equation}\label{eq30}
	\underline{f}(S^*_K)=\left[\left(\sqrt{\hat{f}(S^*_K; {\mathcal{R}}_2)\cdot\theta_2+\frac{2\eta_l}{9}}-\sqrt{\frac{\eta_l}{2}}\right)^2-\frac{\eta_l}{18}\right]\cdot\frac{1}{\theta_2},
\end{equation}
where we have $|\mathcal{R}_2|=\theta_2$, $\eta_l=\ln(1/\delta_l)$, and $\Pr[f(\hat{S}^*_K)>\underline{f}(\hat{S}^*_K)]\geq 1-\delta_l$. Based on Lemma 4.3 in \cite{tang2018online}, the upper bound $\overline{f}(S^\circ_K)$ on $\mathcal{R}_1$ can be modified as
\begin{equation}\label{eq31}
	\overline{f}(S^\circ_K)=\left(\sqrt{{\hat{f}'}(S^\circ_K; {\mathcal{R}}_1)\cdot\theta_1+\frac{\eta_u}{2}}+\sqrt{\frac{\eta_u}{2}}\right)^2\cdot\frac{1}{\theta_1},
\end{equation}
where we have $|\mathcal{R}_1|=\theta_1$, $\eta_u=\ln(1/\delta_u)$, and  $\Pr[f(S^\circ_K)<\overline{f}(S^\circ_K)]\geq 1-\delta_u$. Because the optimal set $S^\circ_K$ is unknown, we use an upper bound ${\hat{f}'}(S^\circ_K; {\mathcal{R}}_1)$ of ${\hat{f}}(S^\circ_K; {\mathcal{R}}_1)$ in Eqn. (\ref{eq31}), which satisfies ${\hat{f}'}(S^\circ_K; {\mathcal{R}}_1)\leq{\hat{f}}(S^*_K; {\mathcal{R}}_1)/(1-1/e)$ according to Ineqn. (\ref{eq26}). Due to its submodularity, this upper bound can be built in the greedy process shown as Algorithm \ref{a2}. Let $\hat{S}^*_a$ be the temporary node set after finishing  the first $a$ iteration in Weighted-Max-Coverage. Based on Lemma 5.1 and 5.2 in \cite{tang2018online}, the upper bound ${\hat{f}'}(S^\circ_K; {\mathcal{R}}_1)$ on $\mathcal{R}_1$ can be modified as ${\hat{f}'}(S^\circ_K; {\mathcal{R}}_1)=$
\begin{equation*}
	\min_{0\leq a\leq K}\left\{\hat{f}(\hat{S}^*_a; {\mathcal{R}}_1)+\sum_{v\in maxMC(\hat{S}^*_a,K; {\mathcal{R}}_1)}\hat{f}(v|\hat{S}^*_a;{\mathcal{R}}_1)\right\},
\end{equation*}
where $maxMC(\hat{S}^*_a,K; {\mathcal{R}}_1)$ is the set of $K$ nodes in $V_R$ with the maximum marginal gain on $\mathcal{R}_1$ with respect to $\hat{S}^*_a$. Finally, we give
\begin{equation}\label{eq32}
	\theta_{max}=\frac{2\left((1-\frac{1}{e})\sqrt{\ln\frac{6}{\delta}}+\sqrt{(1-\frac{1}{e})(\ln\binom{n_R}{K}+\ln\frac{6}{\delta}})\right)^2}{\varepsilon^2\cdot K},
\end{equation}
and $\theta_0=\theta_{max}\cdot\varepsilon^2K$. By setting $i_{max}=\lceil\log_2\frac{\theta_{max}}{\theta_0}\rceil$ and $\delta_l=\delta_u=\delta/(3i_{max})$, the Modified-OPIM-C algorithm can be shown in Algorithm \ref{a3}.

\begin{algorithm}[!t] 
	\caption{{\begin{small}Modified-OPIM-C $(G,T,\vec{W}_T,V_R,K,\varepsilon,\delta)$\end{small}}}
	\label{a3}
	Initialize $\theta_{max}$ by Eqn. (\ref{eq32}) and $\theta_0=\theta_{max}\cdot\varepsilon^2K$\;
	Generate two collections of random MT-RR sets $\mathcal{R}_1$ and $\mathcal{R}_2$ with $|\mathcal{R}_1|=|\mathcal{R}_2|=\theta_0$\;
	$i_{max}=\lceil\log_2\frac{\theta_{max}}{\theta_0}\rceil$\;
	\For{$i=1$ to $i_{max}$}{
		$\hat{S}^*_K\leftarrow$ Weighted-Max-Coverage $(\mathcal{R}_1,V_R,K)$\;
		Compute $\underline{f}(S^*_K)$ on $\mathcal{R}_1$ and $\overline{f}(S^\circ_K)$ on $\mathcal{R}_2$ by Eqn. (\ref{eq30}) and (\ref{eq31}) by setting $\delta_l=\delta_u=\delta/(3i_{max})$\;
		$ratio\leftarrow\underline{f}(S^*_K)/\underline{f}(S^*_K)$\;
		\If{$ratio\geq(1-1/e-\varepsilon)$ or $i=i_{max}$}{
			\Return $\mathcal{R}_1$\;
		}
		Double the size of $\mathcal{R}_1$ and $\mathcal{R}_2$\;
	}
    \Return $\mathcal{R}_1$
\end{algorithm}

\begin{algorithm}[!t] 
	\caption{{\begin{small}Budgeted-W-Max-Coverage $(\mathcal{R}, V_R,\vec{B}, D)$\end{small}}}
	\label{a4}
	Initialize: $S^*\leftarrow 0$, $sum\leftarrow 0$\;
	\While{$S\neq V_R$}{
		Select $v^*\in\arg\max_{v_{k}\in V_{R}\backslash S^*}\hat{f}(v_k|S^*;\mathcal{R})/b_{k}$\;
		\If{$sum+b^*> D$}{
			Break\;
		}
		$S^*\leftarrow S^*\cup\{v^*\}$\;
		$sum\leftarrow sum + b^*$\;
	}
	$v_{max}\leftarrow\arg\max_{v_k\in V_R, b_k\leq D}\hat{f}(\{v_k\};\mathcal{R})$\;
	\Return $\arg\max\{\hat{f}(S^*;\mathcal{R}),\hat{f}(\{v_{max}\};\mathcal{R})\}$\;
\end{algorithm}

\textbf{Summary.} Let $\mathcal{R}^*_1$ be the collection of MT-RR sets returned by Algorithm \ref{a3}. By calling Weighted-Max-Coverage $(\mathcal{R}^*_1,V_R,K)$ shown as Algorithm \ref{a2}, it can return a solution $\hat{S}^*_K$ such that $f(\hat{S}^*_K)\geq(1-1/e-\varepsilon)\cdot f(S^\circ_K)$ with at least $1-\delta$ probability. Back to our MT-DM problem, under the budget $D$, we can use the collection $\mathcal{R}^*_1$ to get a feasible solution, denoted by $\hat{S}^*_D$, by calling Budgeted-W-Max-Coverage $(\mathcal{R}^*_1, V_R, D)$ shown as Algorithm \ref{a4}, where we can omit to compare with the maximum one $v_{max}$ in line 8 of Algorithm \ref{a4} since we have assumed $\max_{v_i\in V_R}\{b_i\}\ll D$. The approximation guarantee of $\hat{S}^*_D$ will be elaborated in the following lemma.
\begin{lem}\label{lem7}
	Let $\mathcal{R}^*_1$ be the collection of MT-RR sets returned by Modified-OPIM-C under the constant $K$ defined in Eqn. (\ref{eq24}). The solution $\hat{S}^*_D$ obtained by greedy process of Budgeted-W-Max-Coverage on $\mathcal{R}^*_1$ under the budget $D$ satisfies
	\begin{equation}\label{eq33}
		f(\hat{S}^*_D)\geq(1-1/\sqrt{e}-\varepsilon)\cdot f(S^\circ_D)
	\end{equation}
	with at least $1-\delta$ proability, where $S^\circ_D$ is the optimal solution under the budget $D$.
\end{lem}
\begin{proof}
	Given a feasible solution $\hat{S}^*_D$, its size $|\hat{S}^*_D|\leq K$ and we have $\binom{n_R}{|\hat{S}^*_D|}\leq\binom{n_R}{K}$ because the $\binom{n_R}{x}$ increases as $x$ increases when $x\leq n_R/2$.  The accumulative error $(\varepsilon',\delta')$ of obtaining $\hat{S}^*_D$ on $\mathcal{R}^*_1$ under the budget $D$ will be less than the error $(\varepsilon,\delta)$ of obtaining $\hat{S}^*_K$ on $\mathcal{R}^*_1$ under the constant $K$. It infers that $\varepsilon'\cdot f(S^\circ_D)\leq\varepsilon\cdot f(S^\circ_K)$ from $\varepsilon'<\varepsilon$ and $f(S^\circ_D)\leq f(S^\circ_K)$. Thus, the Eqn. (\ref{eq33}) has been verified.
\end{proof}

\section{Final Procedure of MT-DM-L}
Based on the multi-task sampling, we can generate a collection of random MT-RR sets whose size is large enough to estimate the Multi-Task Diffusion Function in the MT-DM-L. Following the process of Algorithm \ref{a1}, the final procedure of MT-DM-L can be formulated as follows.
\begin{itemize}
	\item Given a network $G=(V,E)$, a task set $T$, its corresponding Multi-Task IC model $\vec{W}_T$, bidding information $\vec{B}$, a registered user set $V_R\subseteq V$, and budget $D$, it obtains an integer $K$ by following Eqn. (\ref{eq24}).
	\item Given error parameters $(\varepsilon,\delta)$, it generate a collection of random MT-RR sets $\mathcal{R}$ by calling Modified-OPIM-C $(G,T,\vec{W}_T, V_R,K,\varepsilon,\delta)$ shown as Algorithm \ref{a3}.
	\item It runs the MT-DM-L shown as Algorithm \ref{a1} based on $\mathcal{R}$, where it uses $\mathcal{R}$ to estimate the marginal gain in MT-DM-L. Specifically, it replaces the line 3 of Algorithm \ref{a1} with ``$v_{i+1}\in\arg\max_{v_k\in V_R\backslash S_i}\hat{f}(v_k|S_i; \mathcal{R})/b_k$'' and replaces the line 14 of Algorithm \ref{a1} with ``$v_{i_{l+1}}\in\arg\max_{v_{i_k}\in V_{R:-i}\backslash H_l}\hat{f}(v_{i_k}|H_l;\mathcal{R})/b_{i_k}$''. Like this, it can get a feasible winning worker set $S$ and pricing $\vec{P}$.
\end{itemize}

Here, the collection of random MT-RR sets $\mathcal{R}$ defined above is sufficient to get an approximate solution in the winning worker selection stage, while it is also enough to determine the pricing for each winning user $v_i\in S$ in the payment determination stage. Under the same budget $D$, the size of the candidate set $V_{R:-i}$ is less than $V_R$. Now, we can draw the main conclusion of this paper.

\textbf{Running time analysis: } According to the above procedure, the total running time can be separated by the generation of the collection $\mathcal{R}$ and running of M-TDM based on $\mathcal{R}$. Based on Lemma 6.2 in \cite{tang2018online}, the expected number of random MT-RR sets in $\mathcal{R}$, denoted by $\mathbb{E}[|\mathcal{R}|]$ can be modified as
\begin{equation}
	\mathbb{E}[|\mathcal{R}|]=\mathcal{O}((K\ln n+\ln(1/\delta))\cdot\varepsilon^{-2}/f(S^\circ_K)).
\end{equation}
For any MT-RR set $R\in\mathcal{R}$, its size should be $|R|\leq f(S^\circ_K)$. Thus, the running time of generating $\mathcal{R}$ is bounded by $\mathcal{O}((K\ln n+\ln(1/\delta))\cdot\varepsilon^{-2})$. Given a collection $\mathcal{R}$, the expected running time to execute Weight-Max-Coverage $(\mathcal{R},V_R,K)$ shown as Algorithm \ref{a2} is $\mathcal{O}(\sum_{R\in\mathcal{R}}|R|)=\mathcal{O}((K\ln n+\ln(1/\delta))\cdot\varepsilon^{-2})$. Then, in the MT-DM-L, after the winning worker selection stage, the size of the winning worker set $S$ must be less than $K$. For each winning worker $v_i\in S$, it needs to greedily select another at most $K$ users for determining its pricing. The running time of this entire process must be less than that of executing Weight-Max-Coverage $(\mathcal{R},V_R,K)$ $K+1$ times. Thus, the total running time can be bounded by $\mathcal{O}((K+1)(K\ln n+\ln(1/\delta))\cdot\varepsilon^{-2})$.

\begin{thm}[Main Theorem]
	The MT-DM-L given by Algorithm \ref{a1} based on the collection of random MT-RR sets returned by Modified-OPIM-C is an effective incentive mechanism, which can satisfy individual rationality, truthfulness, computational efficiency, and $(1-1/\sqrt{e}-\varepsilon)$ approximation with at least $1-\delta$ probability in $\mathcal{O}((K+1)(K\ln n+\ln(1/\delta))\cdot\varepsilon^{-2})$ expected running time.
\end{thm}
\begin{proof}
	It can be easily induced from the procedure of MT-DM-L by combining Lemma \ref{lem1}, Lemma \ref{lem2}, Lemma \ref{lem3}, Lemma \ref{lem7}, and the above running time analysis.
\end{proof}

\section{Evaluation and Performance}
In this section, we conduct several experiments to evaluate the performance of our proposed algorithms and mechanisms. All of our simulations are programmed by Python 3.8 and run on a Mac M1 machine.
\subsection{Simulation Setup}
Our evaluations are based on two real-world social networks, Damascus \cite{damascus} and Dash \cite{dash}. They are re-tweeting networks, where nodes are Twitter users and edges are retweets. These were collected from various social and political hashtags. Since our MT-DM problem is defined on a location-dependent social network, we create a simulation environment in an area $A_G$ with $1000\times1000$, in which each block with $100\times 100$ is a subarea. Thus, there are a total of $100$ subareas in this area. For each user $v_i$ in the social network, it has a coordinate $(x_i,y_i)$, where $x_i$ and $y_i$ are uniformly sampled in range $[0,1000]$. Thus, we can easily calculate which subarea the user belongs to.

We consider task sets with different numbers of tasks, where the number of tasks ranges from $1$ to $4$, namely there are at most four tasks in our setting, denoted by $T=\{t_1,t_2,t_3,t_4\}$. For each task $t_j\in T$, it expected completion quality $q_j^k$ in subarea $a_k$ is uniformly sampled in $[0,1]$. And for each task $t_j\in T$, we need to define its diffusion parameter set $W^j$. We assume that the diffusion probability of each task in its corresponding diffusion layer is constant. Thus, we give the parameter sets of the Multi-Task IC model as $\vec{W}_T=\{W^1=0.3, W^2=0.5, W^3=0.4, W^4=0.4\}$, where it indicates that $w_{uv}^1=0.3$ for all $(u,v)\in E$ in task $t_1$, and so on. Here, we do not use the Weighted Cascade \cite{guo2020multi} to set the diffusion probability for convenience in comparison. Then, the number of registered users in $V_R$ is $10\%$ to $40\%$ of the total number of users. Back to the auction model, each user $v_i\in V_R$ will provide its bidding information $B_i=\{T_i,b_i\}$. Here, each registered user claims each task in the published task set with $50\%$ probability, and its bid is equal to the number of its claimed tasks multiplied by its unit bid, where the unit bid for each user is uniformly sampled in $[0.8,1.2]$.

Next, the performance evaluation will be divided into two parts. The first part is to validate that our Modified-OPIM-C algorithm can return a feasible solution to maximize the Multi-Task Diffusion Function under the budget, which is similar to studying the IM problem. Once we can efficiently get the approximate solution, we validate the effectiveness of our incentive mechanism MT-DM-L in the second part.

\subsection{Diffusion Function Estimation}
We introduce some typical baselines, which will be used for comparison with our Modified-OPIM-C algorithm. For the sake of fair comparison, the standard value of the Multi-Task Diffusion Function $f(S)$ can be given as the result obtained by using $2000$ Monto Carlo (MC) simulations. After a seed set is returned by baselines, the standard value computed by the above method is taken as its experimental result.
\begin{itemize}
    \item \textbf{Greedy:} It adopts the weighted hill-climbing algorithm \cite{khuller1999budgeted} that selects the user with maximum marginal gain $f(v_k|S_i)/b_k$ in each iteration, where the Multi-Task Diffusion Function is estimated by using $500$ MC simulations.
    \item \textbf{Modified-OPIM-C:} Let $\mathcal{R}$ be the collection of random MT-RR sets returned by Algorithm \ref{a3} by setting $\varepsilon=0.1$ and $\epsilon=0.1$, then it calls Budgeted-W-Max-Coverage shown as Algorithm \ref{a4} based on $\mathcal{R}$.
    \item \textbf{Greedy-IM:} It adopts the hill-climbing algorithm \cite{kempe2003maximizing} (similar to solving the IM problem by using influence spread instead of Multi-Task Diffusion Function) that selects the user with maximum marginal gain $\sigma(v_k|S_i)/(b_k/|T_k|)$ in each iteration, where the diffusion probability is the average among all tasks and the influence function is estimated by using $500$ MC simulations.
    \item \textbf{OPIM-C:} It adopts the OPIM-C algorithm \cite{tang2018online} (similar to solving the IM problem) where the diffusion probability is the average among all tasks.
    \item \textbf{MaxDegree:} It selects the user with maximum out-degree $N^+(v_k)/(b_k/|T_k|)$ in each iteration.
    \item \textbf{Random:} It randomly selects a feasible registered user in each iteration.
\end{itemize}

\begin{figure}[!t]
	\centering
	\subfigure[$RU=20\%$, $|T|=2$]{
		\includegraphics[width=0.48\linewidth]{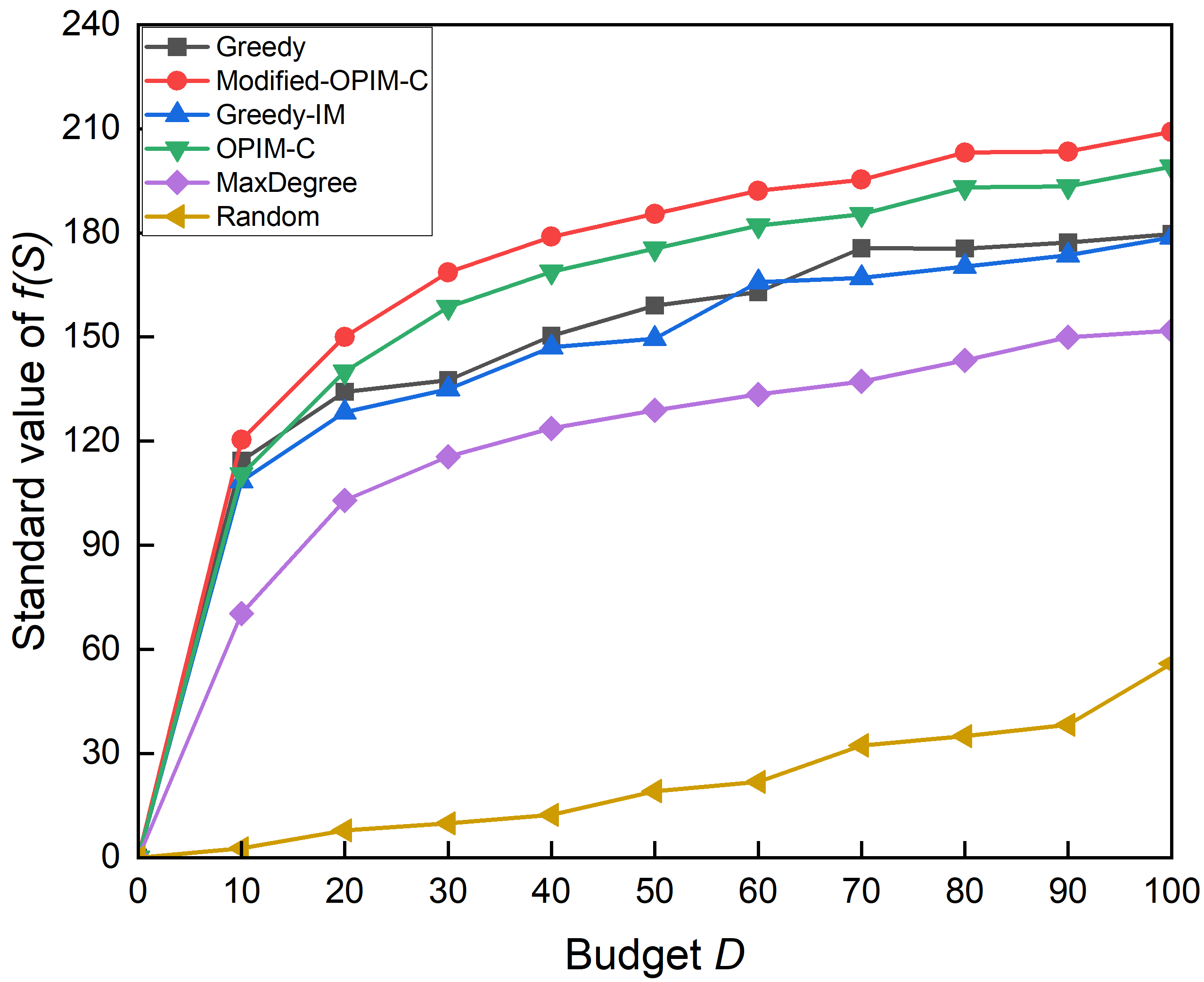}
		%\caption{fig1}
	}%
	\subfigure[$RU=40\%$, $|T|=2$]{
		\includegraphics[width=0.48\linewidth]{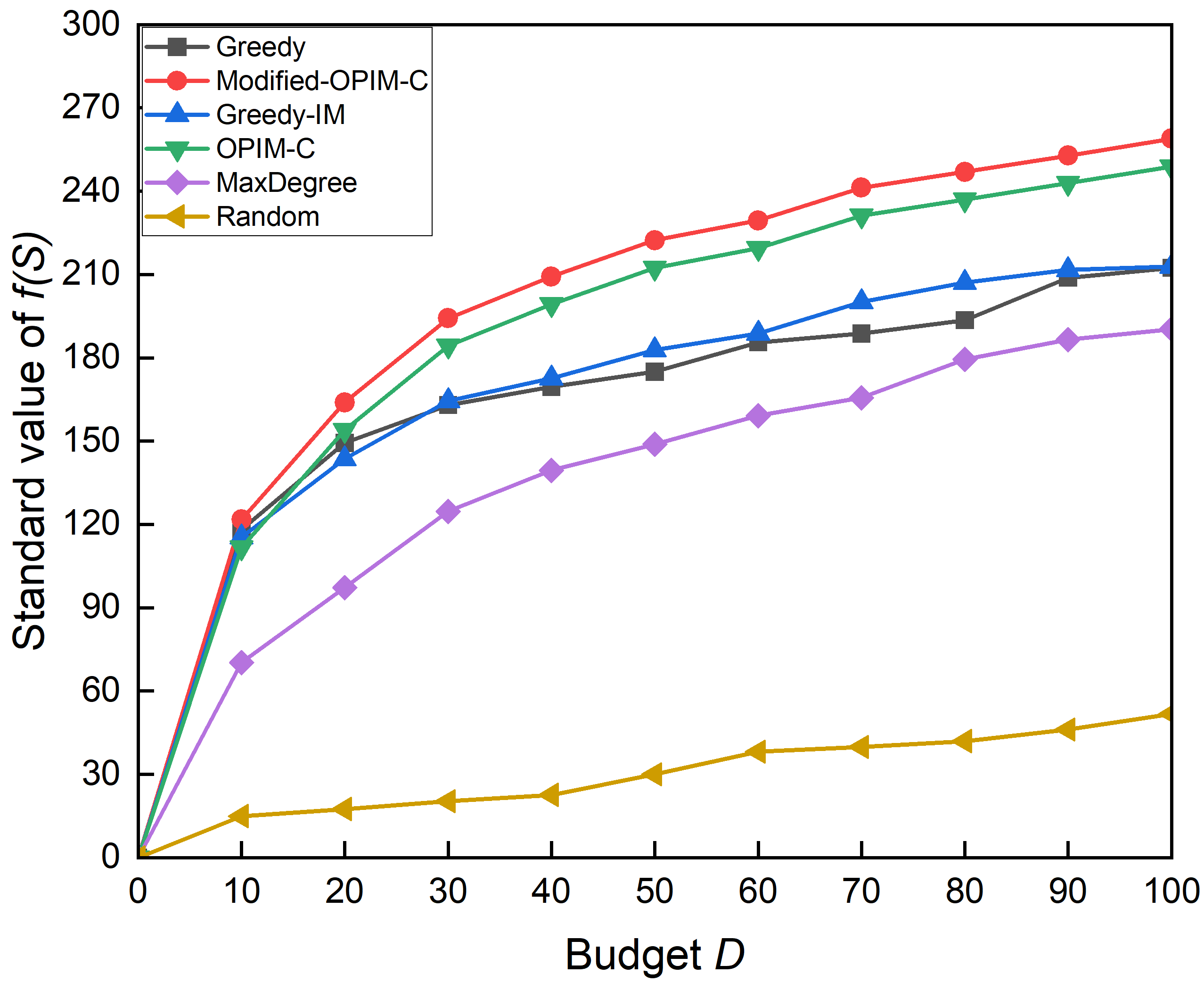}
		%\caption{fig1}
	}%
	\centering
	
	\subfigure[$RU=20\%$, $|T|=3$]{
		\includegraphics[width=0.48\linewidth]{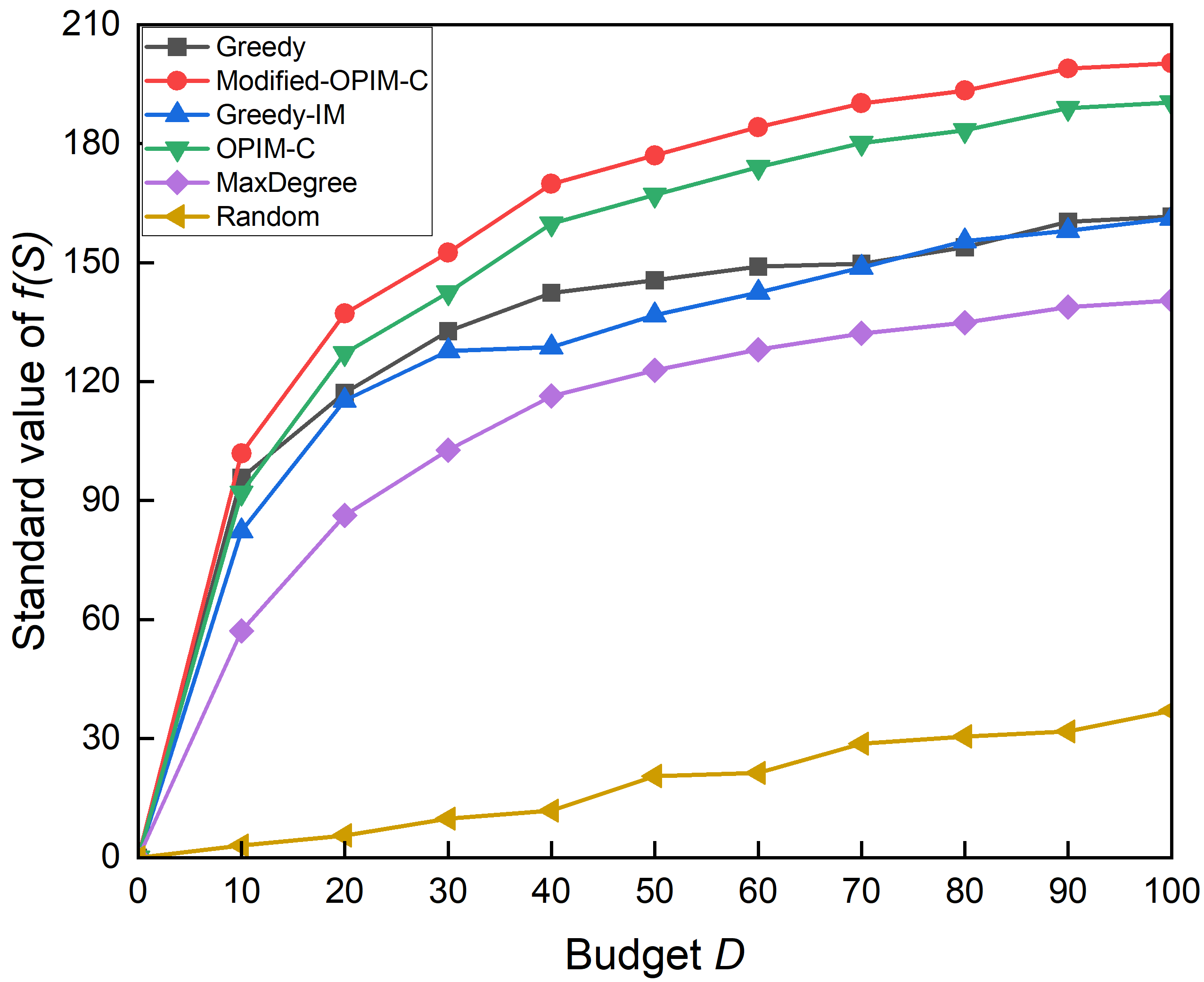}
		%\caption{fig1}
	}%
	\subfigure[$RU=40\%$, $|T|=3$]{
		\includegraphics[width=0.48\linewidth]{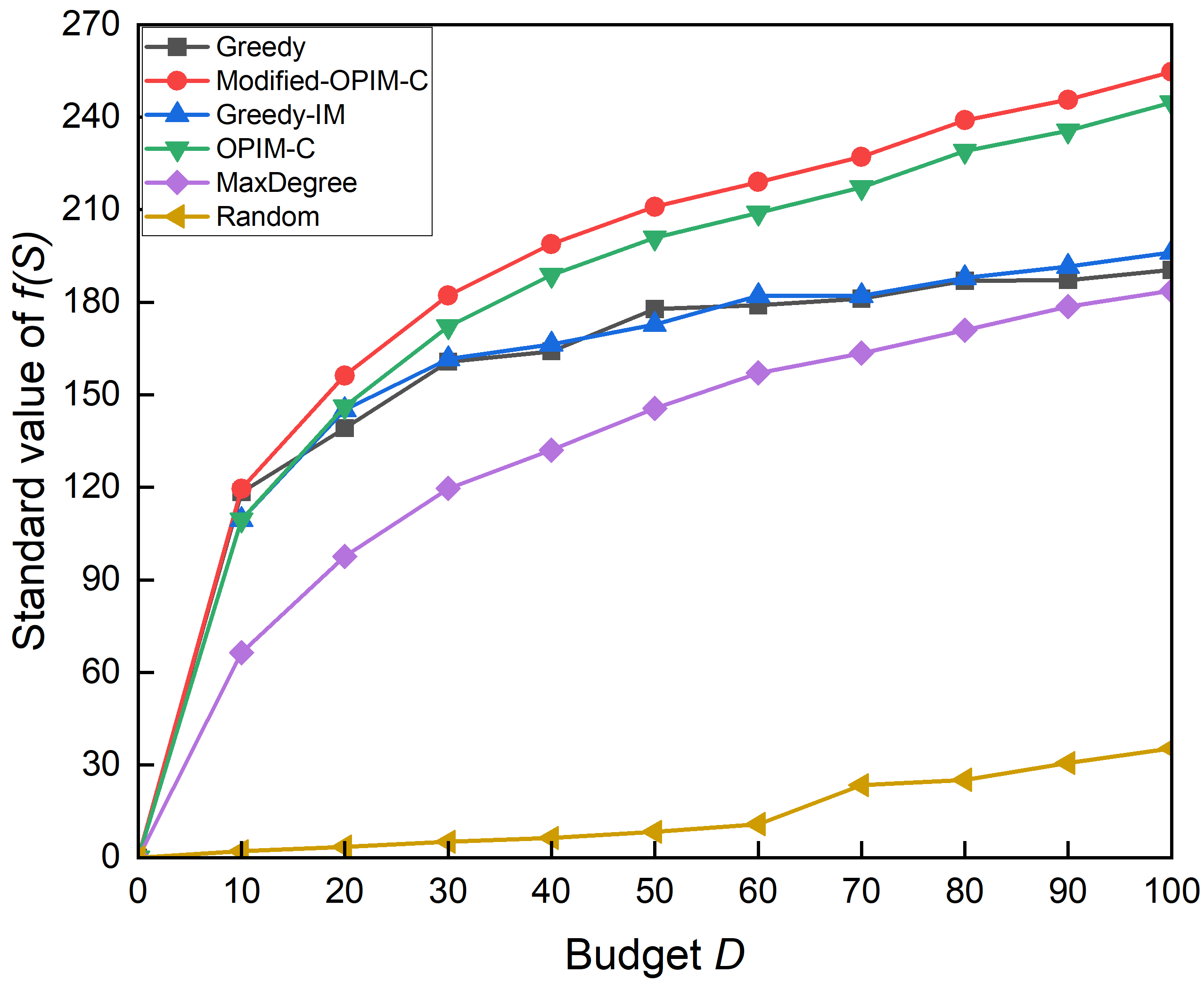}
		%\caption{fig1}
	}%
	\centering
	\caption{The performance comparison achieved under the Damascus dataset by giving different parameter settings.}
	\label{fig2}
\end{figure}

\subsubsection{Performance} Fig. \ref{fig2} and Fig. \ref{fig3} draw the performance comparison achieved by all baselines under the Damascus and Dash datasets by giving different parameter settings, where ``$RU=20\%, |T|=2$'' means the proportion of registered users is $20\%$ and the number of tasks is $2$. In this part, Greedy and Greedy-IM are simulation-based algorithms, which are mainly used to compare with their corresponding sampling-based algorithms: Modified-OPIM-C and OPIM-C. Here, we have several observations. First, in all datasets and parameter settings, Modified-OPIM-C always achieves the best performance, which validates the effectiveness of our proposed algorithm to maximize the multi-task diffusion under the budget. Second, these greedy-based algorithms, including Greedy, Modified-OPIM-C, Greedy-IM, and OPIM-C, are far better than heuristic algorithms like MaxDegree and Random. It shows the correctness of our incentive mechanism design based on a greedy selection strategy. 

Third, the gap between Modified-OPIM-C and Greedy is bigger than we expected because in fact they all implement the same procedure, but one is based on MC simulations and the other is based on sampling. Thus, it is expected that they achieve similar results. A reasonable explanation is that $500$ MC simulations are too few, which makes it difficult for Greedy to accurately find the user with maximum marginal gain in each iteration. It leads to the selection of suboptimal users, thereby causing error accumulation and final result deterioration. In fact, we also verify this opinion in the follow-up experiments: with the increase of MC simulations, the performance of Greedy has also been significantly improved, and it approaches that of Modified-OPIM-C. Fourth, even though the gap between OPIM-C and Greedy-IM also exists, this gap is obviously smaller than the gap between Modified-OPIM-C and Greedy. We think this is also caused by the error of MC simulations. Even though both Greedy and Greedy-IM adopt $500$ MC simulations, this is more accurate to estimate the influence spread than our Multi-Task Diffusion Function, because our function involves multiple tasks and is more complex, thus the estimation error is greater and the probability of bad events is higher. This further illustrates the necessity of designing sampling-based algorithms.

\begin{figure}[!t]
	\centering
	\subfigure[$RU=10\%$, $|T|=2$]{
		\includegraphics[width=0.48\linewidth]{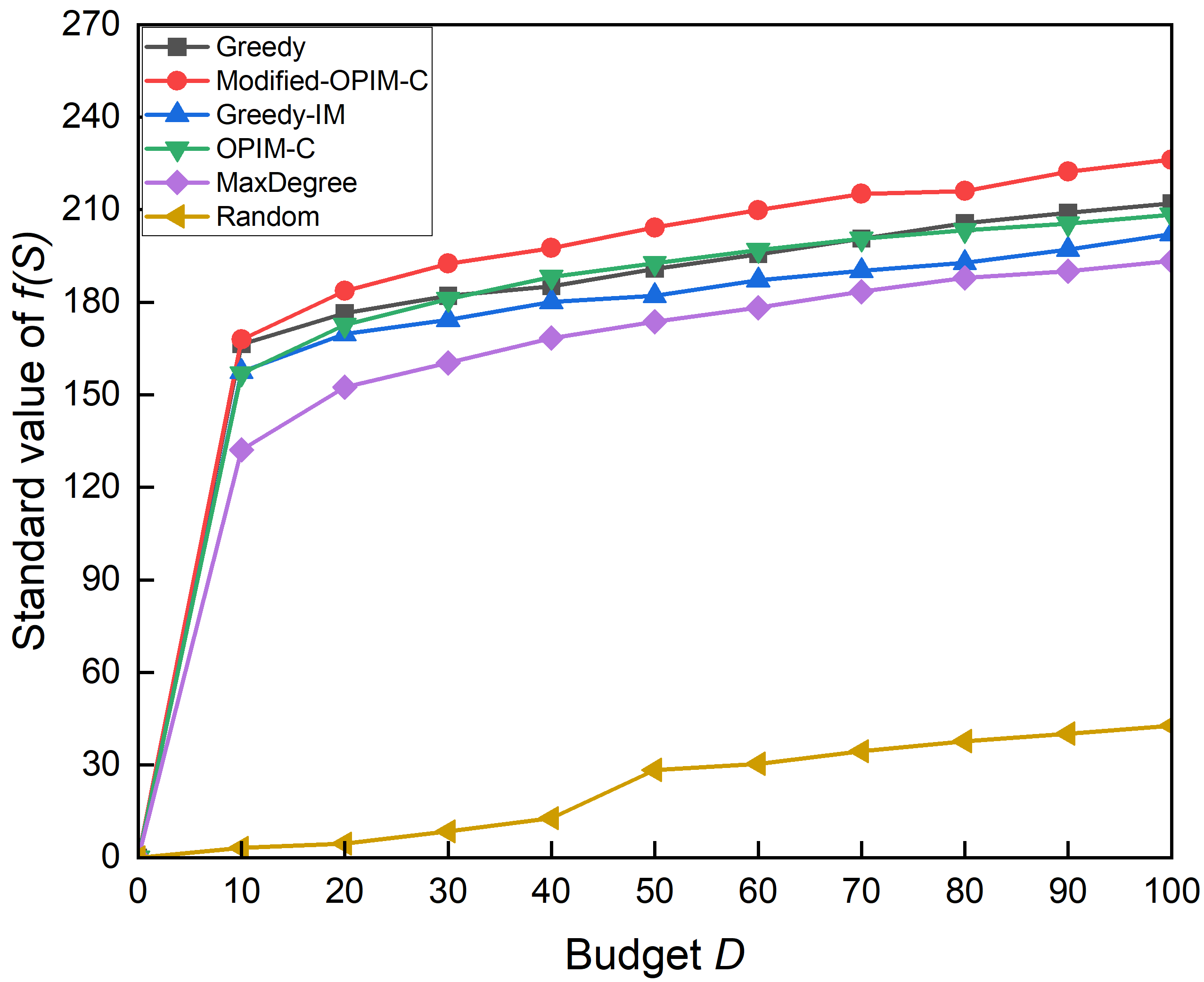}
		%\caption{fig1}
	}%
	\subfigure[$RU=20\%$, $|T|=2$]{
		\includegraphics[width=0.48\linewidth]{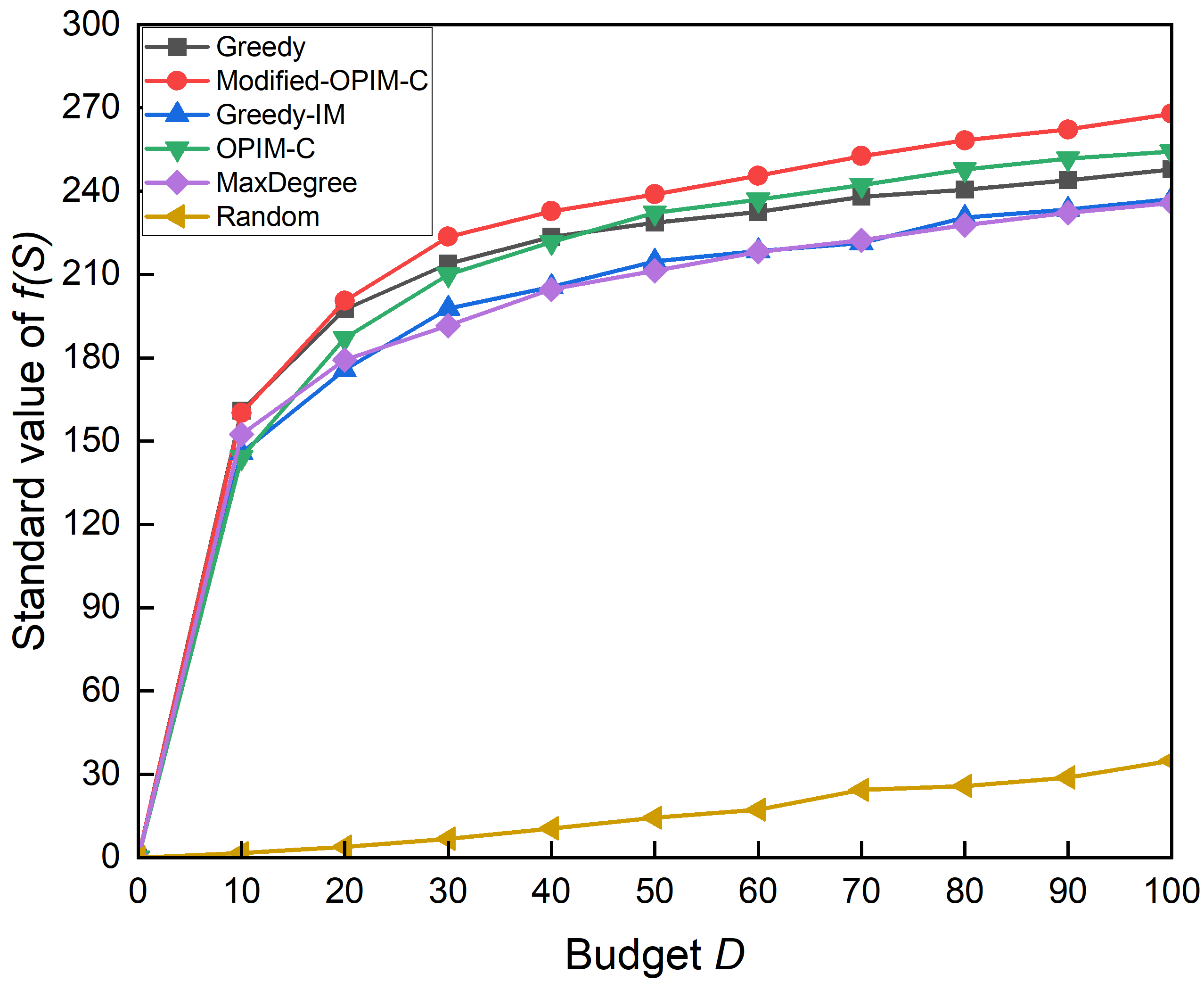}
		%\caption{fig1}
	}%
	\centering
	
	\subfigure[$RU=10\%$, $|T|=3$]{
		\includegraphics[width=0.48\linewidth]{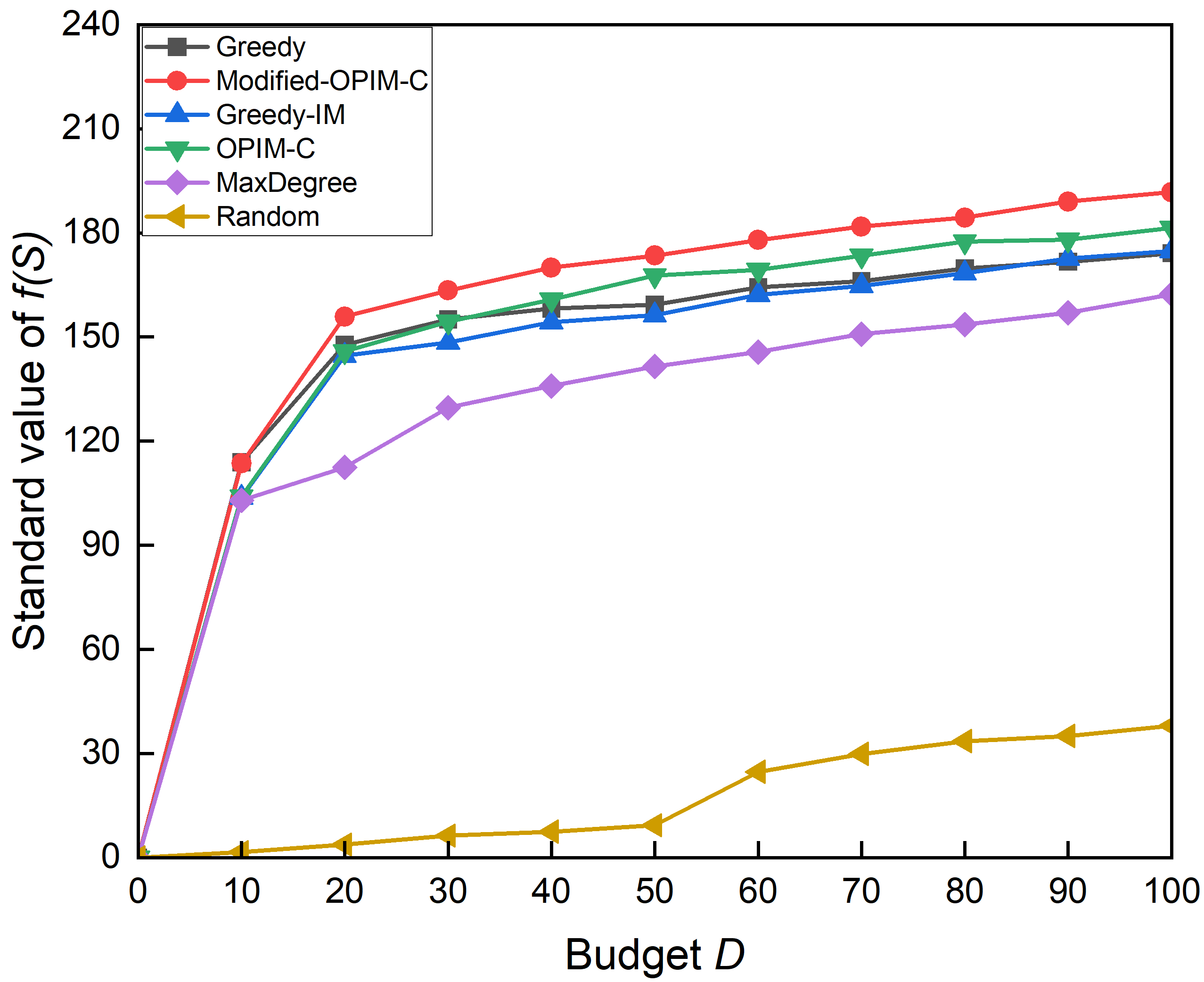}
		%\caption{fig1}
	}%
	\subfigure[$RU=20\%$, $|T|=3$]{
		\includegraphics[width=0.48\linewidth]{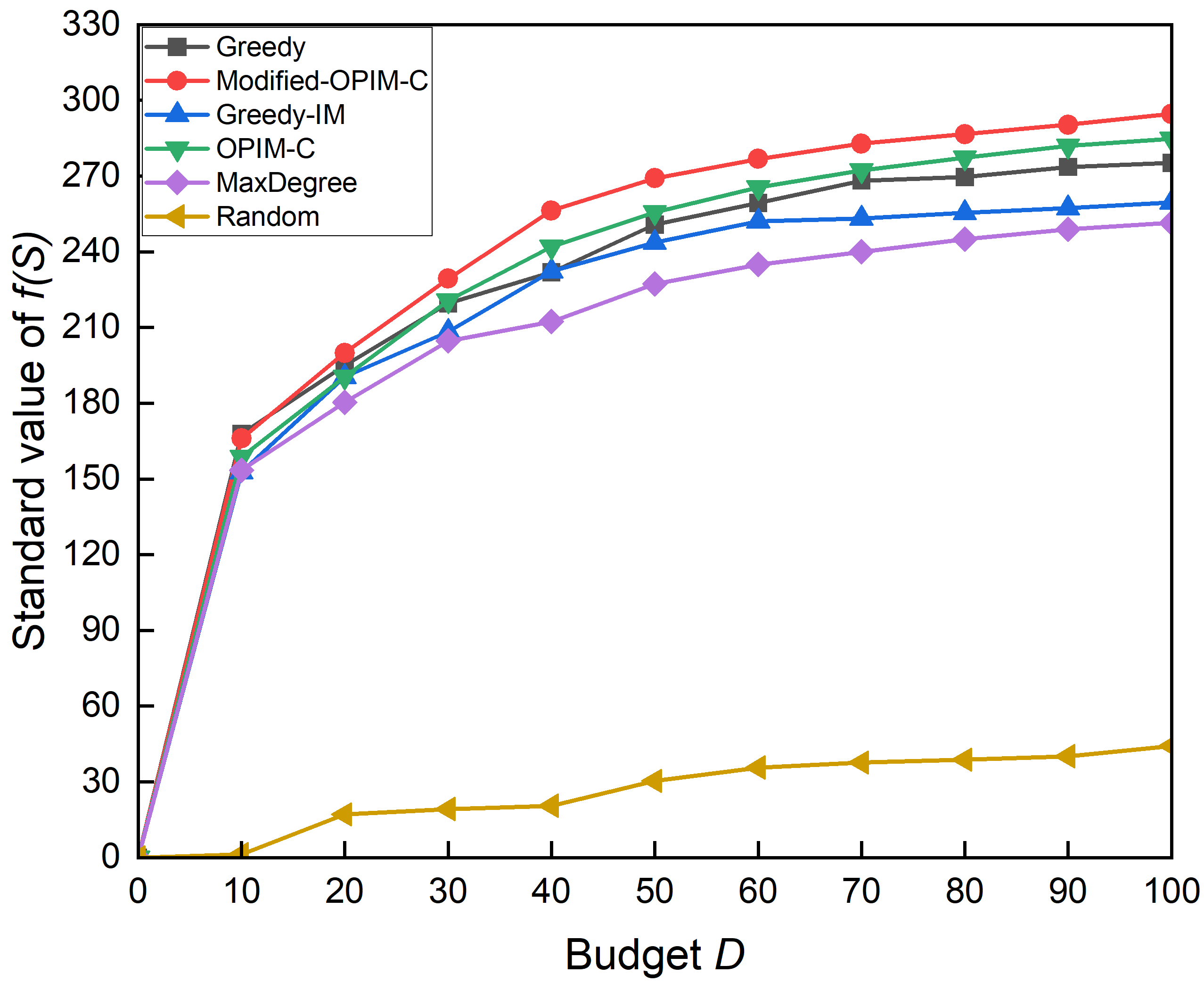}
		%\caption{fig1}
	}%
	\centering
	\caption{The performance comparison achieved under the Dash dataset by giving different parameter settings.}
	\label{fig3}
\end{figure}

\begin{figure}[!t]
	\centering
	\subfigure[Damascus]{
		\includegraphics[width=0.48\linewidth]{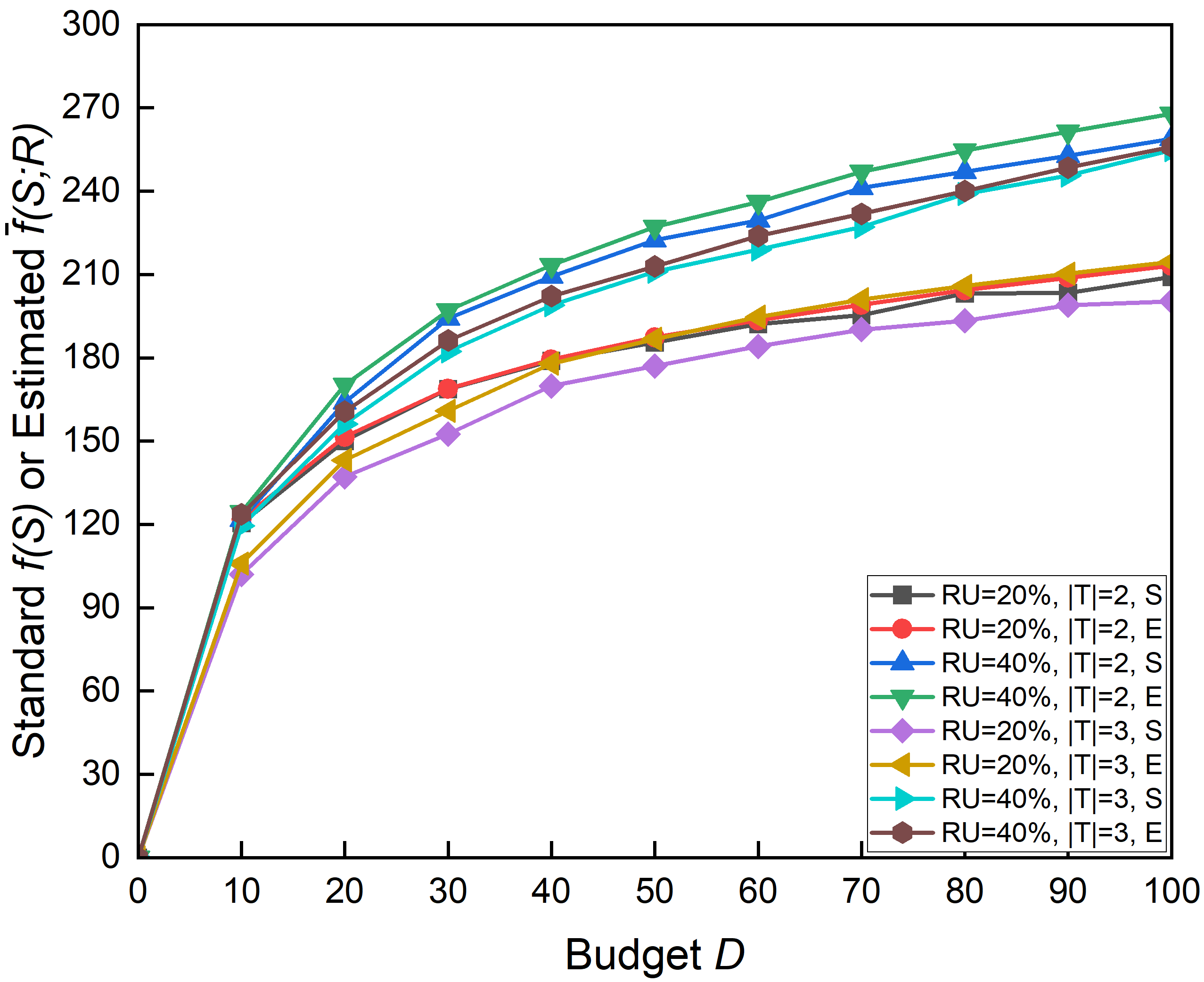}
		%\caption{fig1}
	}%
	\subfigure[Dash]{
		\includegraphics[width=0.48\linewidth]{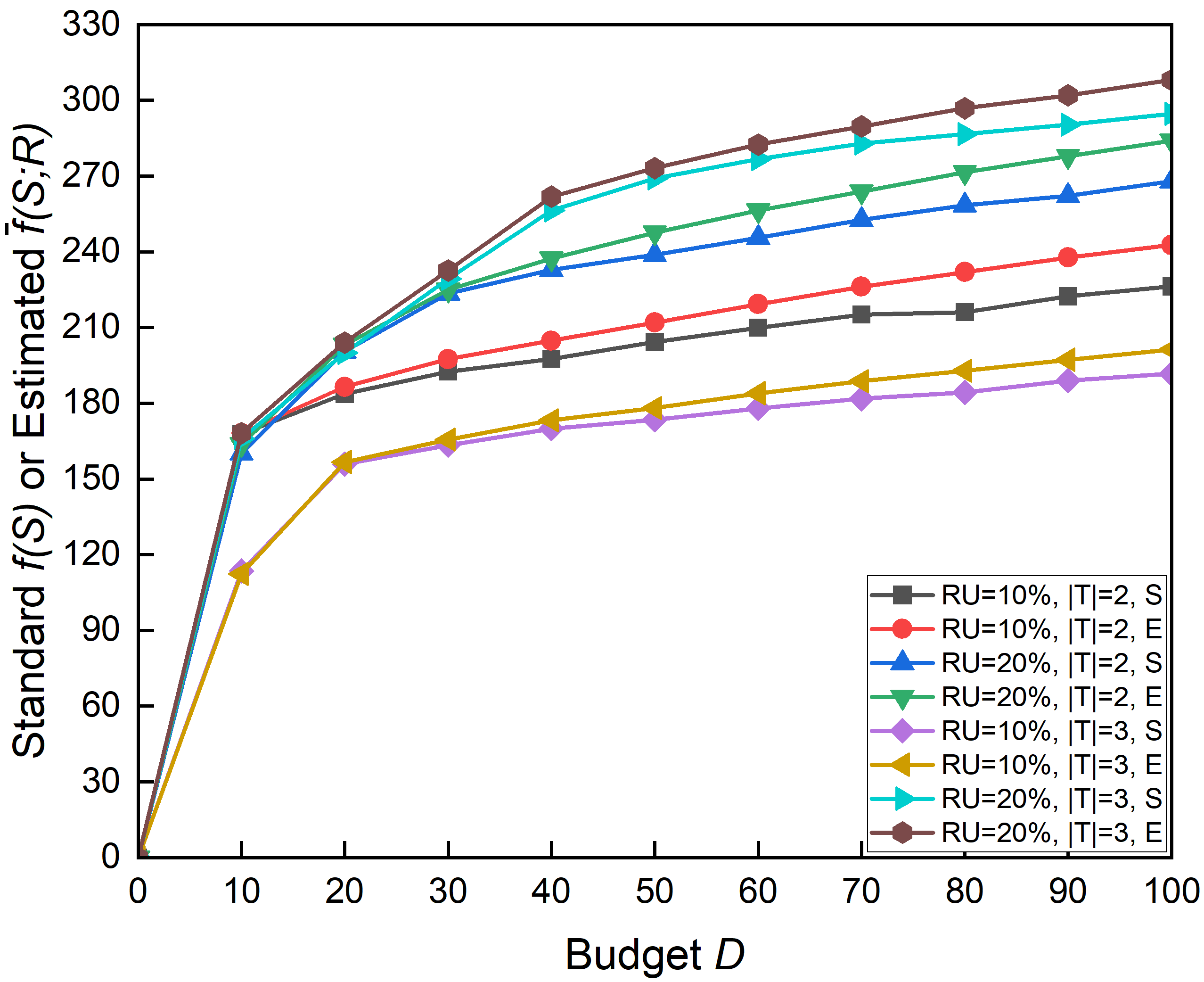}
		%\caption{fig1}
	}%
	\centering
	\caption{The gap between the standard value and estimated value.}
	\label{fig4}
\end{figure}

\subsubsection{Sampling and error analysis} In the Modified-OPIM-C, it greedily selects a seed set $S^*$ based on the collection of random MT-RR sets $\mathcal{R}$ returned by Algorithm \ref{a3} under the budget. Thus, we can compare its estimated value $\hat{f}(S^*;\mathcal{R})$ with the stardard value $f(S^*)$ obtained by $2000$ MC simulations. Fig. \ref{fig4} draws the gap between the standard value and estimated value of the Multi-Task Diffusion Function under two datasets by giving different parameter settings. Here, we can see that the standard value and estimated value are generally close, and the error is controlled within a very small range. This shows that our multi-task sampling can accurately estimate the objective function. In addition, the estimated value is slightly larger than the standard value, and as the budget increases, the gap between the standard value and estimated value is increasing. This is because the seed set is selected based on $\mathcal{R}$, thus it is overestimated to some extent, and the error is accumulated with the increase of budget.

\begin{figure}[!t]
	\centering
	\subfigure[$RU=10\%$, $|T|=2$]{
		\includegraphics[width=0.48\linewidth]{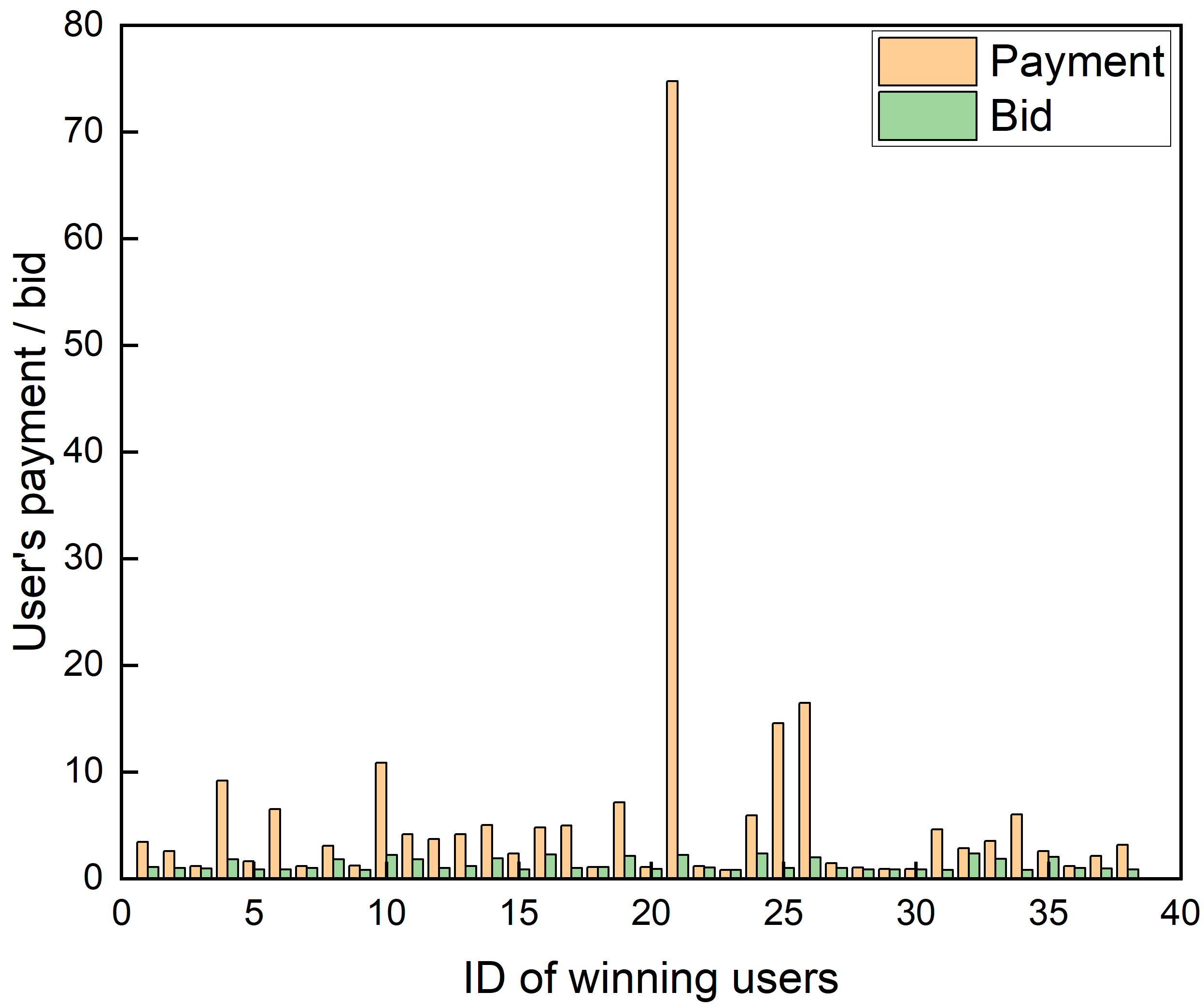}
		%\caption{fig1}
	}%
	\subfigure[$RU=20\%$, $|T|=2$]{
		\includegraphics[width=0.48\linewidth]{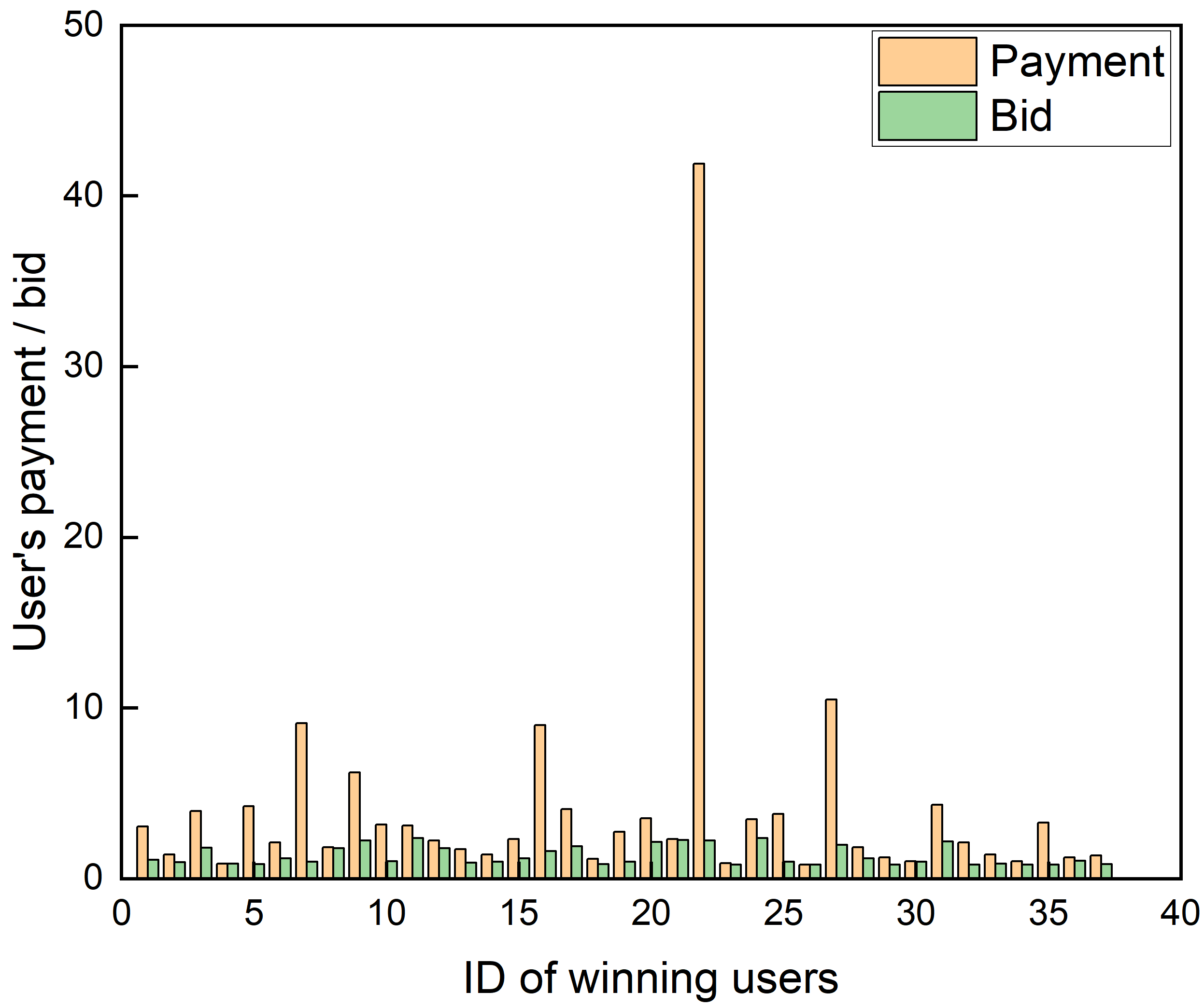}
		%\caption{fig1}
	}%
	\centering
	
	\subfigure[$RU=10\%$, $|T|=3$]{
		\includegraphics[width=0.48\linewidth]{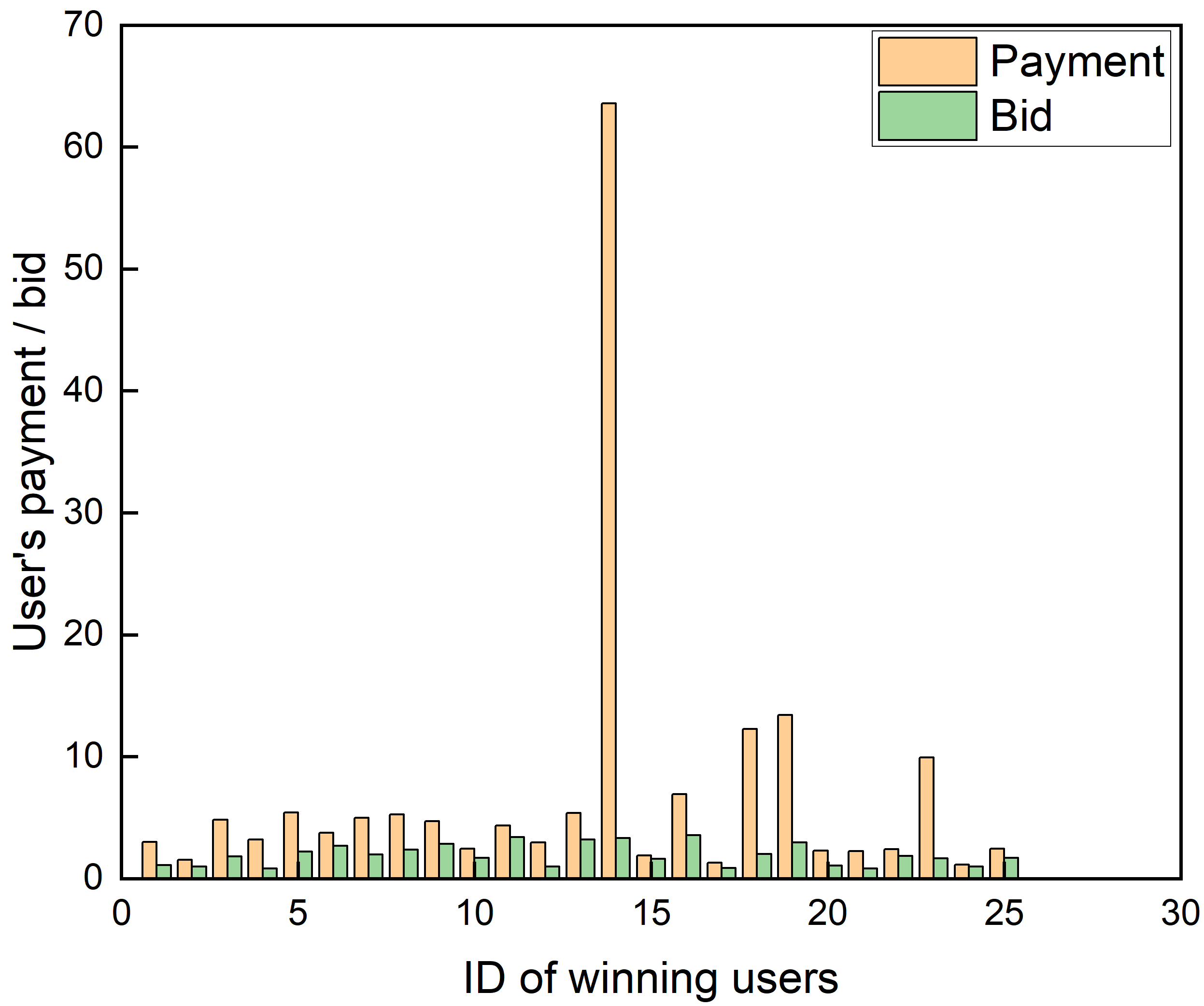}
		%\caption{fig1}
	}%
	\subfigure[$RU=20\%$, $|T|=3$]{
		\includegraphics[width=0.48\linewidth]{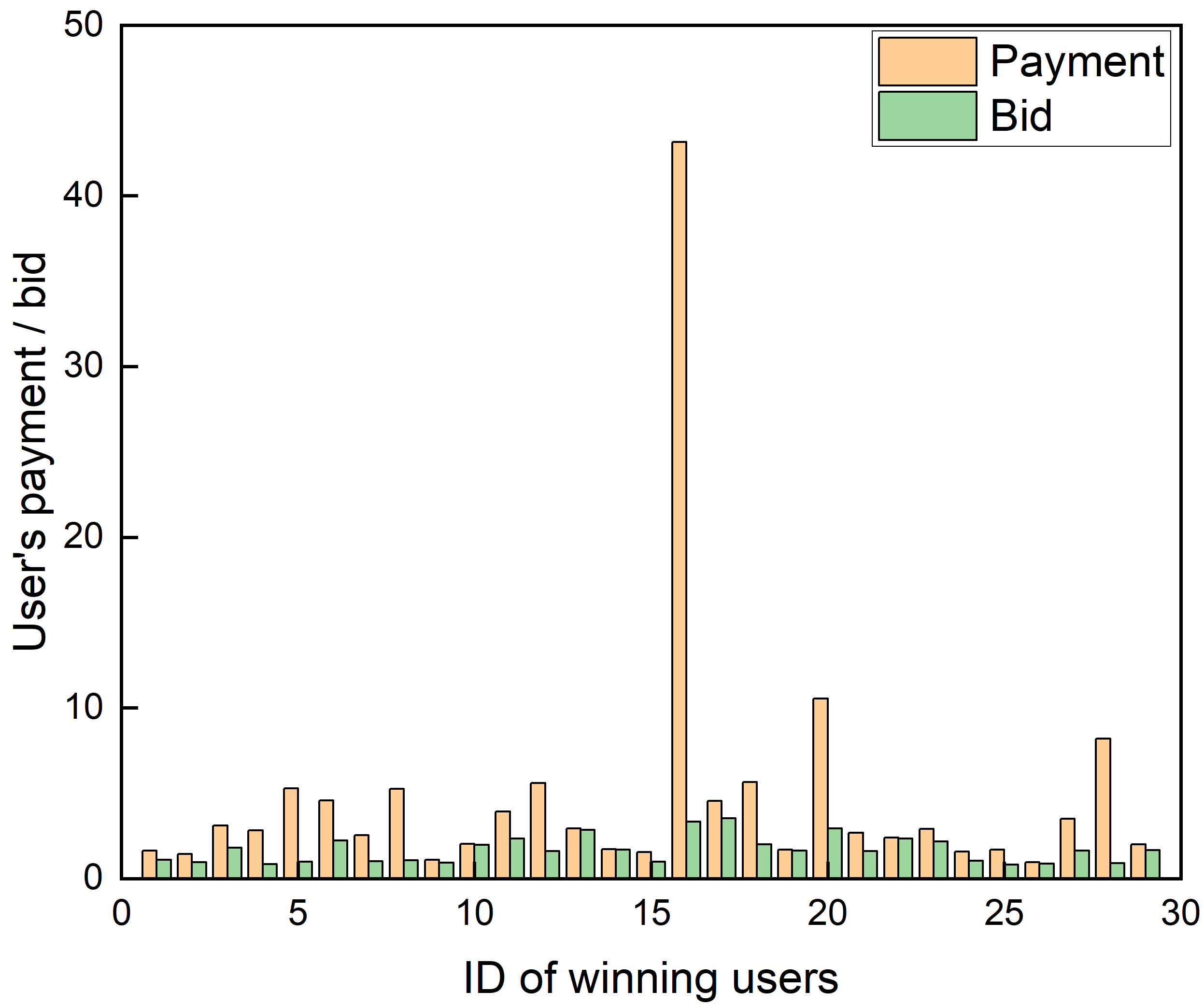}
		%\caption{fig1}
	}%
	\centering
	\caption{The payment and bid of winning users achieved by MT-DM-L under the Damascus dataset and budget $D=50$.}
	\label{fig5}
\end{figure}

\begin{table}[!t]
	\renewcommand{\arraystretch}{1.1}
	\caption{The running time comparison for $D=100$, where $s=second$, $m=minute$, and $h=hour$}
	\label{table1}
	\centering
	\resizebox{\linewidth}{!}{
		\begin{tabular}{|c|c|c|c|c|c|}
			\hline
			\bfseries Dataset & \bfseries Algorithm & \bfseries (a) & \bfseries (b) & \bfseries (c) & \bfseries (d)\\
			\hline
			\multirow{2}*{Dama} & Greedy & 2.87h & 7.38h & 4.39h & 11.2h\\
			\cline{2-6}
			~ & Mo-OPIM-C & 148s & 250s & 141s & 213s\\
			\hline
			\multirow{2}*{Dash} & Greedy & 3.76h & 10.6h & 5.10h & 18.13h\\
			\cline{2-6}
			~ & Mo-OPIM-C & 174s & 317s & 154s & 319s\\
			\hline
	\end{tabular}}
\end{table}

\subsubsection{Running time}
Table \ref{table1} shows the running time for $D=100$ under different parameter settings, where the settings of (a), (b), (c), and (d) correspond to the serial numbers in the Fig. \ref{fig2} and Fig. \ref{fig3}. Here, we only consider the approximate algorithms, Greedy and Modified-OPIM-C. First, the computational cost of simulation-based algorithm, Greedy, is much higher than its corresponding sampling-based algorithm, Modified-OPIM-C. Second, as the scale of the problem increases, including the increase in the number of tasks and registered users, the running time of Greedy will obviously deteriorate. However, the running time of Modified-OPIM-C only increases slightly. Thus, simulation-based algorithms will not be appropriate to use in the following incentive mechanism, and we will use Modified-OPIM-C to estimate the Multi-Task Diffusion Function in the MT-DM-L, whose running time is acceptable.

\begin{figure}[!t]
	\centering
	\subfigure[$RU=10\%$, $|T|=2$]{
		\includegraphics[width=0.48\linewidth]{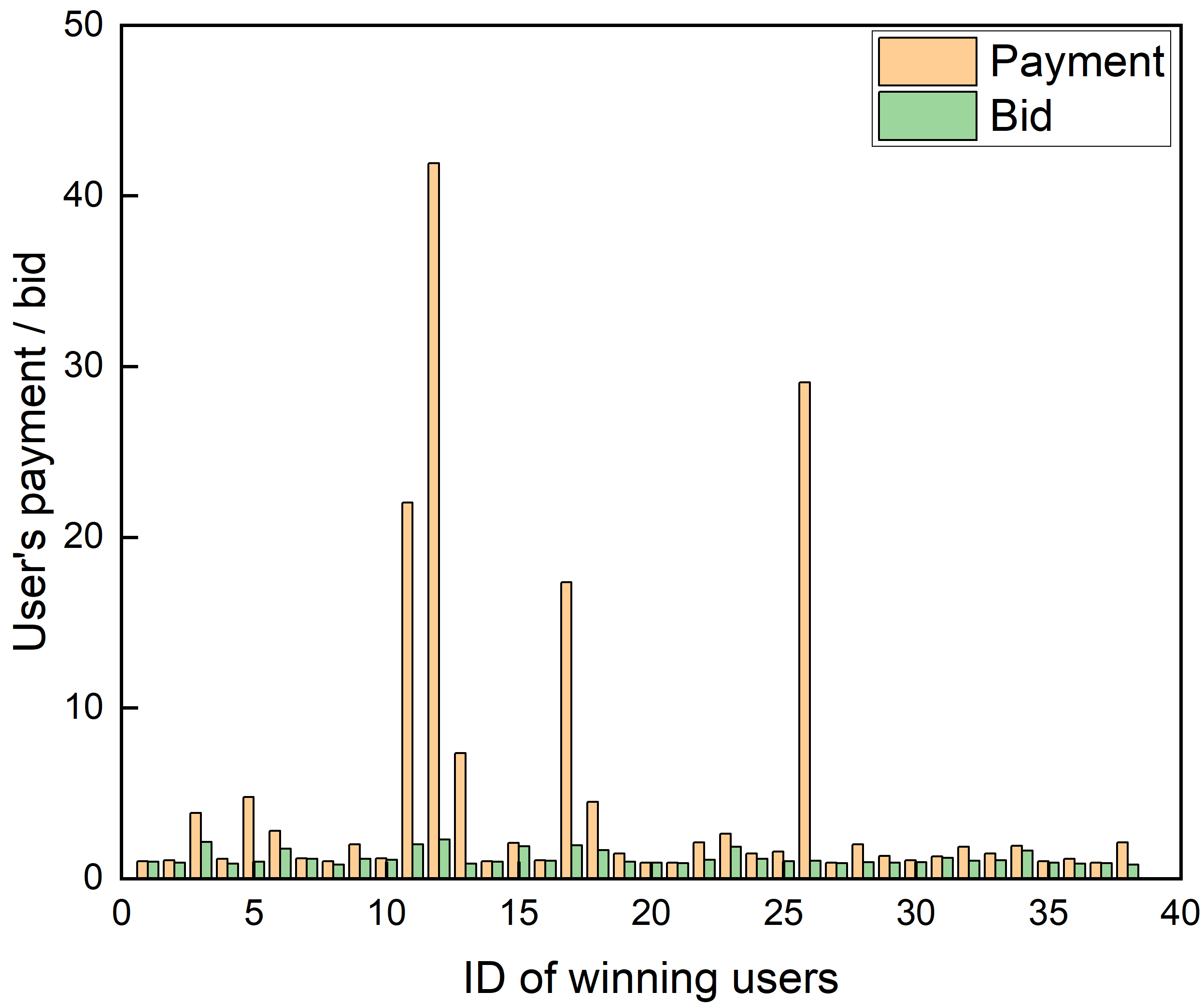}
		%\caption{fig1}
	}%
	\subfigure[$RU=20\%$, $|T|=2$]{
		\includegraphics[width=0.48\linewidth]{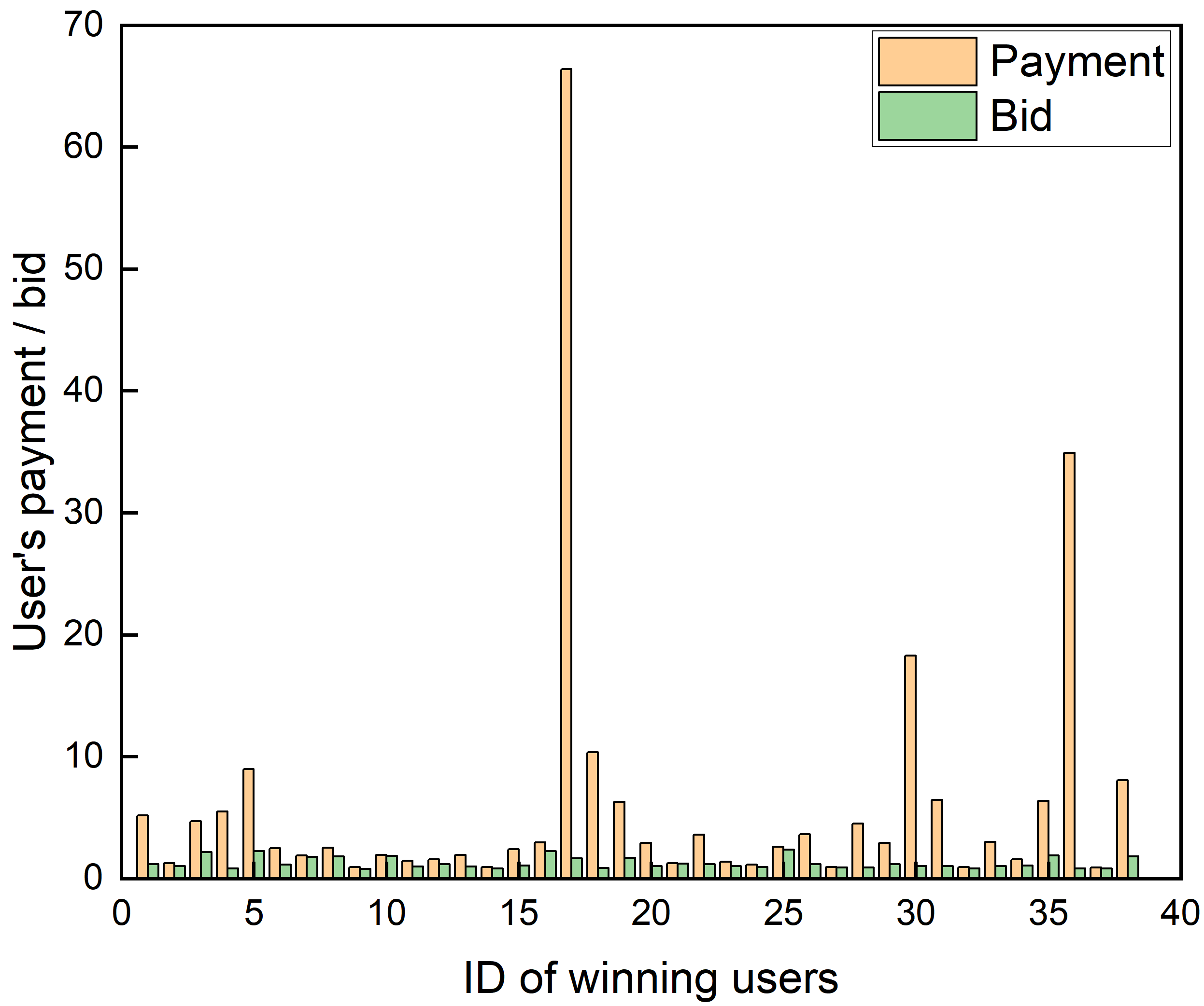}
		%\caption{fig1}
	}%
	\centering
	
	\subfigure[$RU=10\%$, $|T|=3$]{
		\includegraphics[width=0.48\linewidth]{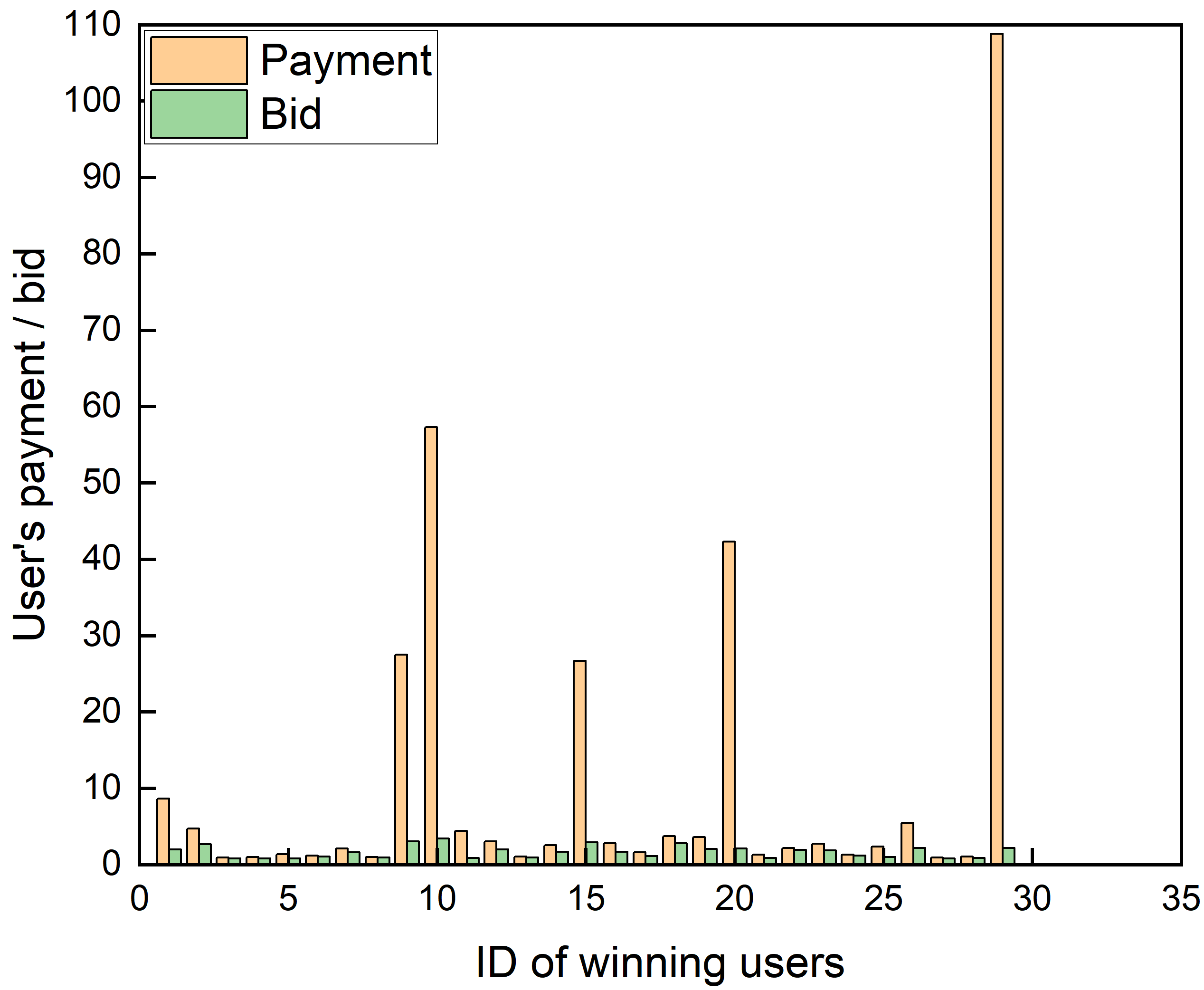}
		%\caption{fig1}
	}%
	\subfigure[$RU=20\%$, $|T|=3$]{
		\includegraphics[width=0.48\linewidth]{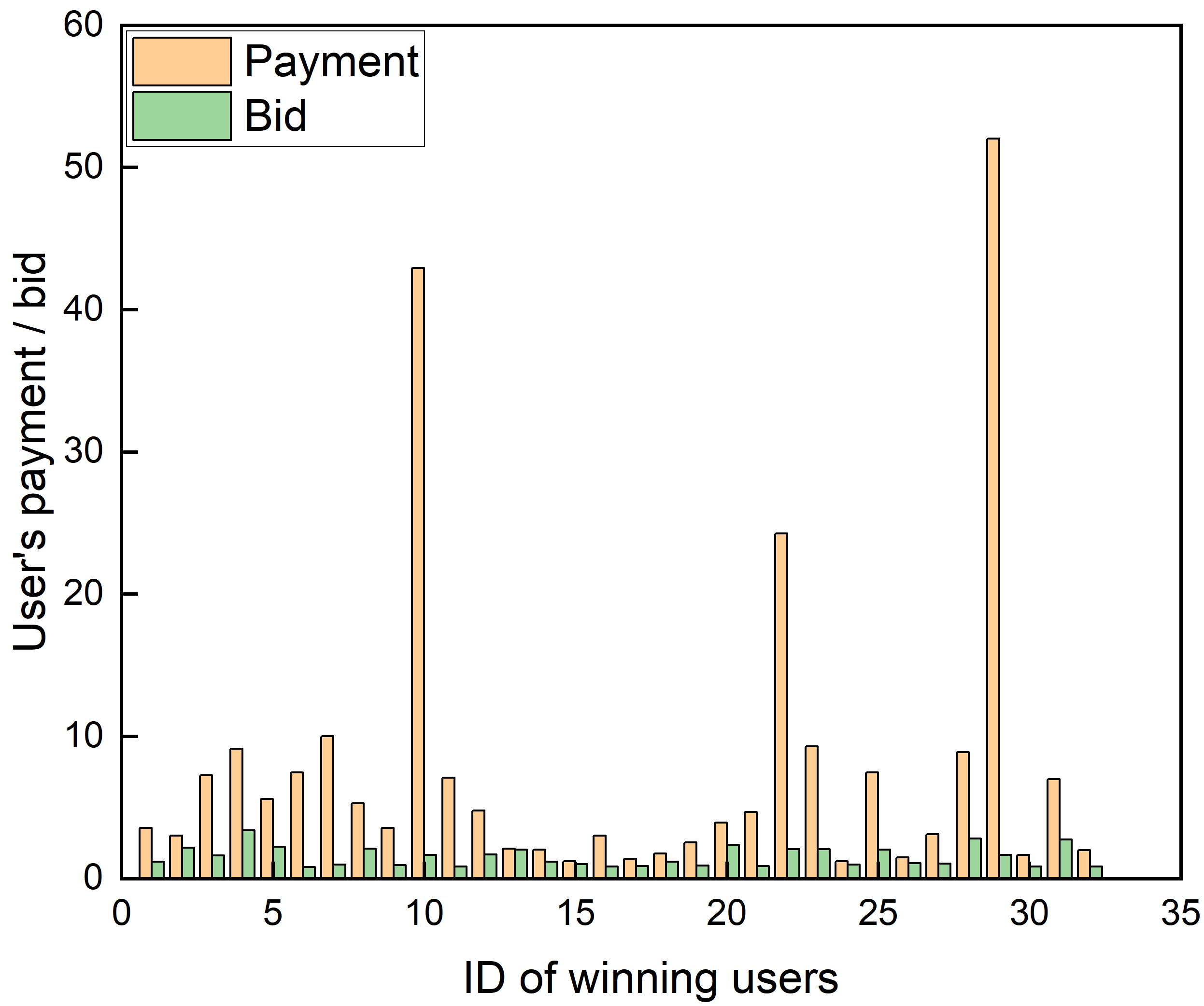}
		%\caption{fig1}
	}%
	\centering
	\caption{The payment and bid of winning users achieved by MT-DM-L under the Dash dataset and budget $D=50$.}
	\label{fig6}
\end{figure}

\begin{table}[!t]
	\renewcommand{\arraystretch}{1.1}
	\caption{The overpayment ratio achieved by MT-DM-L under two datasets and budget $D=50$}
	\label{table2}
	\centering
		\begin{tabular}{|c|c|c|c|c|}
			\hline
			\bfseries Dataset & \bfseries (a) & \bfseries (b) & \bfseries (c) & \bfseries (d)\\
			\hline
			{Dama} & 3.4744 & 2.0437 & 2.5499 & 1.8019\\
			\hline
			{Dash} & 4.9825 & 3.8208 & 5.7473 &  4.1122\\
			\hline
	\end{tabular}
\end{table}

\subsection{Incentive Mechanism}
To implement the incentive mechanism MT-DM-L shown as Algorithm \ref{a1}, we apply the Modified-OPIM-C shown as Algorithm \ref{a3} to generate a collection of random MT-RR sets, and use it to estimate the Multi-Task Diffusion Function in MT-DM-L. This is because Modified-OPIM-C achieves the best performance with the least running time. Fig. \ref{fig5} and Fig. \ref{fig6} draw the payment and bid of winning users under the Damascus and Dash datasets by giving budget $D=50$ and different parameter settings. Here, we have several observations. First, for each winning user, its payment is larger than its bid, thus the utility of every user is positive. We can say that they are individually rational. Second, since we have relaxed the constraint from $\sum_{v_i\in S}p_i\leq D$ to $\sum_{v_i\in S}b_i\leq D$, we need to quantify the degree of violating the budget balance. Given a seed set $S$, we define the overpayment ratio $OR(S)$ as
\begin{equation}
	OR(S)=\left(\sum\nolimits_{v_i\in S}p_i-\sum\nolimits_{v_i\in S}b_i\right)/\sum\nolimits_{v_i\in S}b_i.
\end{equation}
Table \ref{table2} shows the overpayment ratio achieved by MT-DM-L under different parameter settings, where the settings of (a), (b), (c), and (d) correspond to the serial numbers in the Fig. \ref{fig5} and Fig. \ref{fig6}. Comparing (a) with (b) (or (c) with (d)), we find that the overpayment ratio decreases as the number of registered users increases. This is because the competition becomes more intense with the increase in registered users, and the critical price becomes lower. In general, overpayment ratios are not as expected (too high), which is related to the uniformity of our parameter settings. No matter the users with large influence or small influence, their unit bids are uniformly sampled in $[0.8, 1.2]$, which is the main reason why some influential users get very high payments. In practical applications, influential users usually bid more, thus this extreme situation will be avoided.

\begin{figure}[!t]
	\centering
	\subfigure[A winning user]{
		\includegraphics[width=0.48\linewidth]{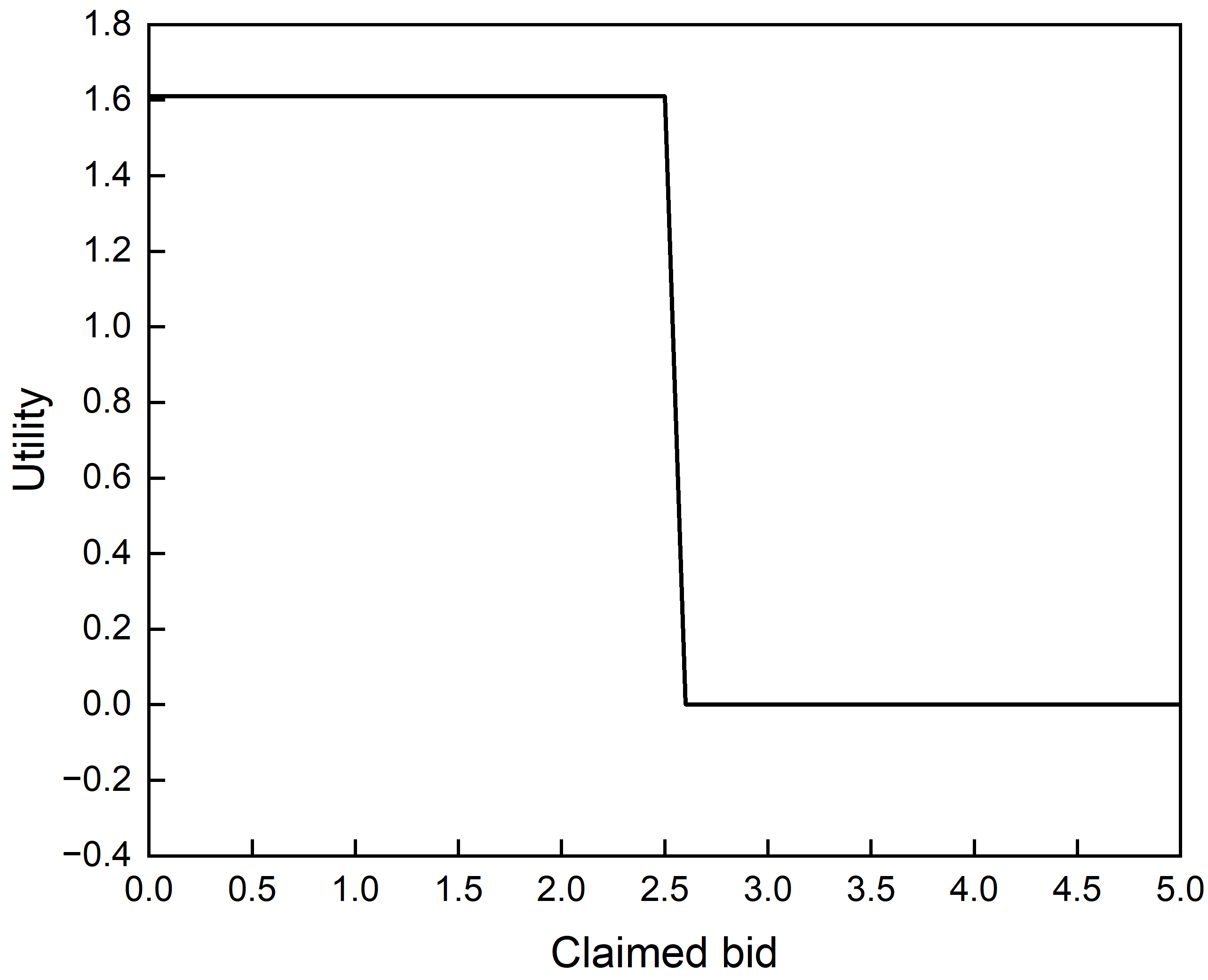}
		%\caption{fig1}
	}%
	\subfigure[A losing user]{
		\includegraphics[width=0.48\linewidth]{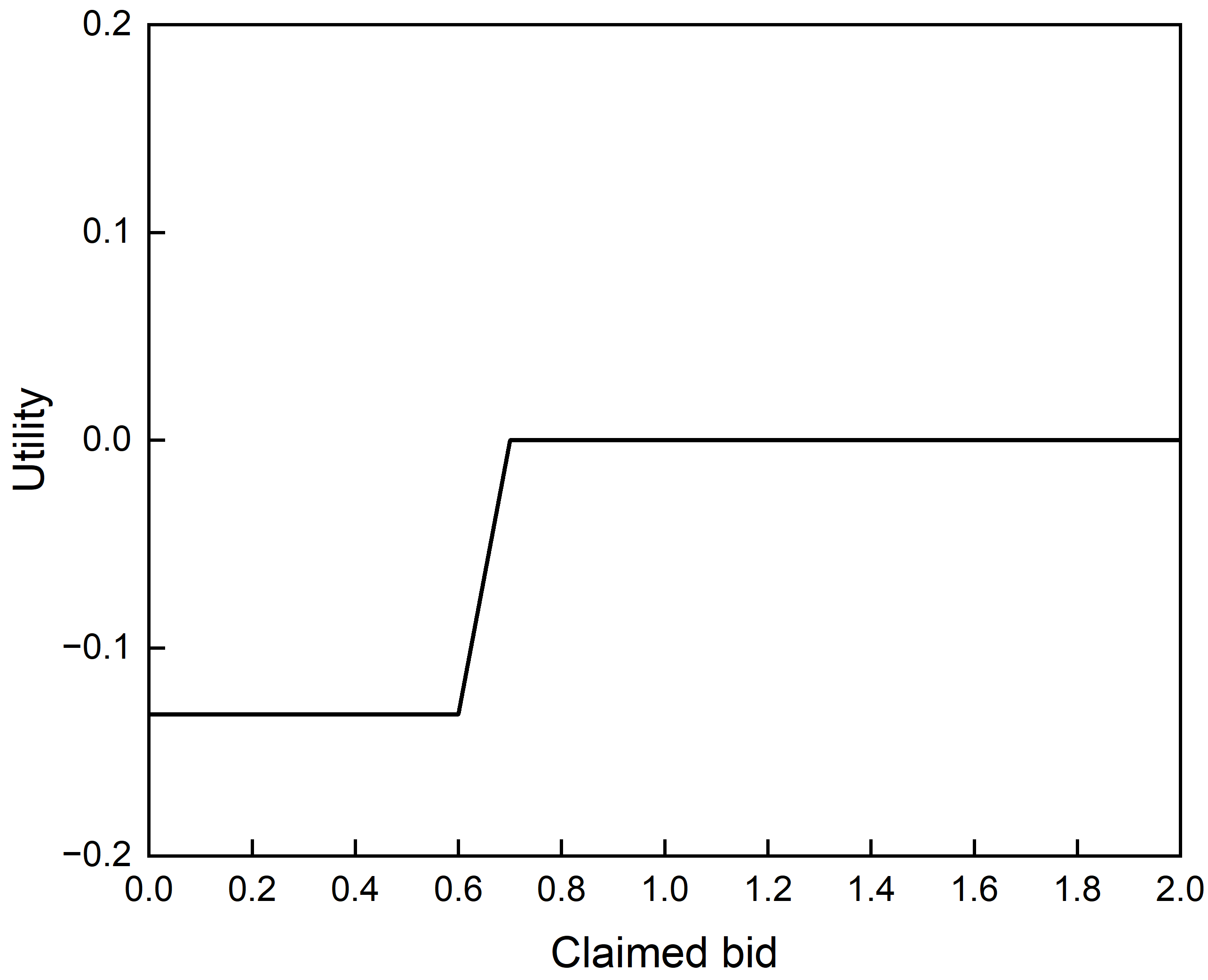}
		%\caption{fig1}
	}%
	\centering
	\caption{A representative user in Damascus with $RU=20\%$ and $T=2$.}
	\label{fig7}
\end{figure}

\begin{figure}[!t]
	\centering
	\subfigure[A winning user]{
		\includegraphics[width=0.48\linewidth]{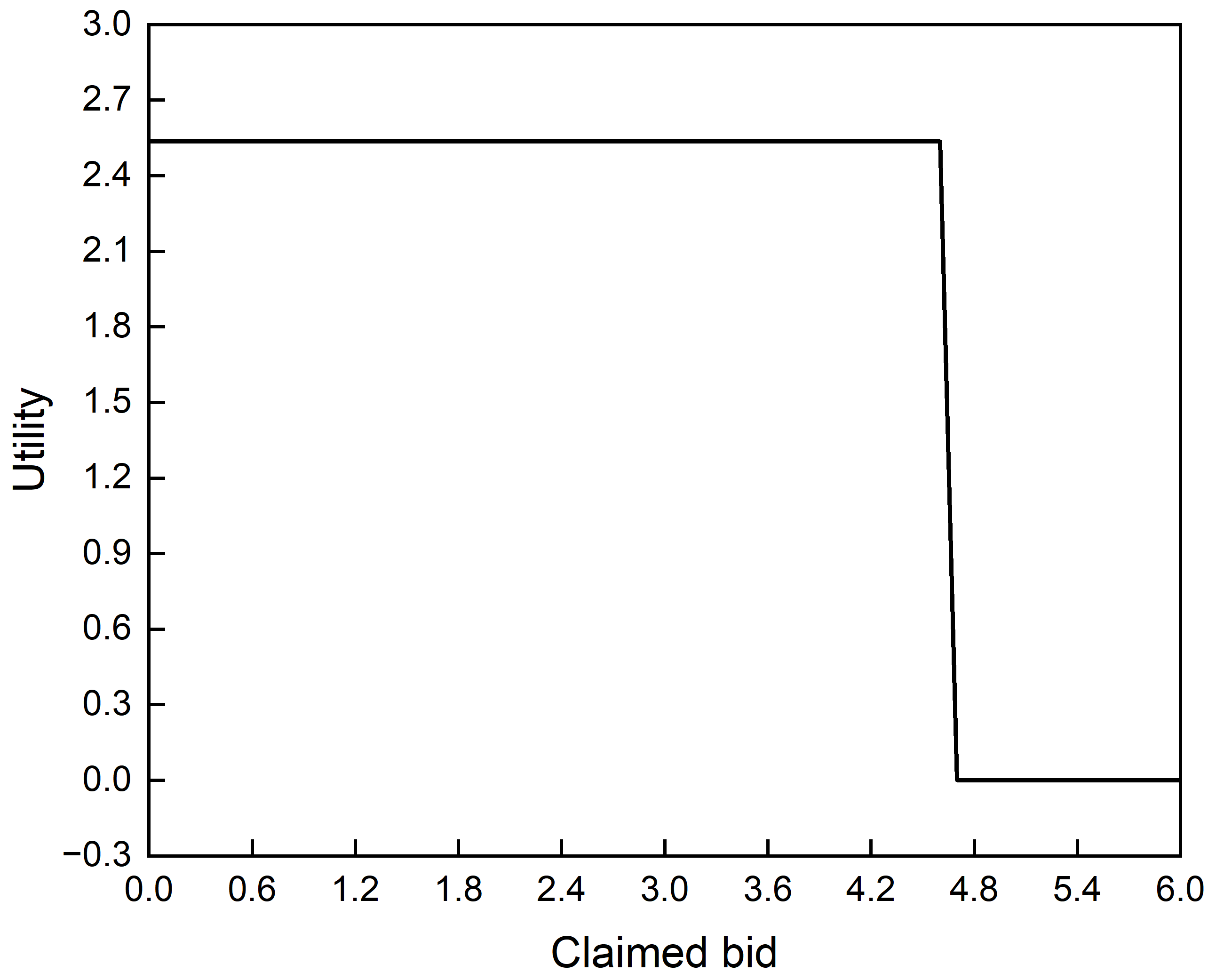}
		%\caption{fig1}
	}%
	\subfigure[A losing user]{
		\includegraphics[width=0.48\linewidth]{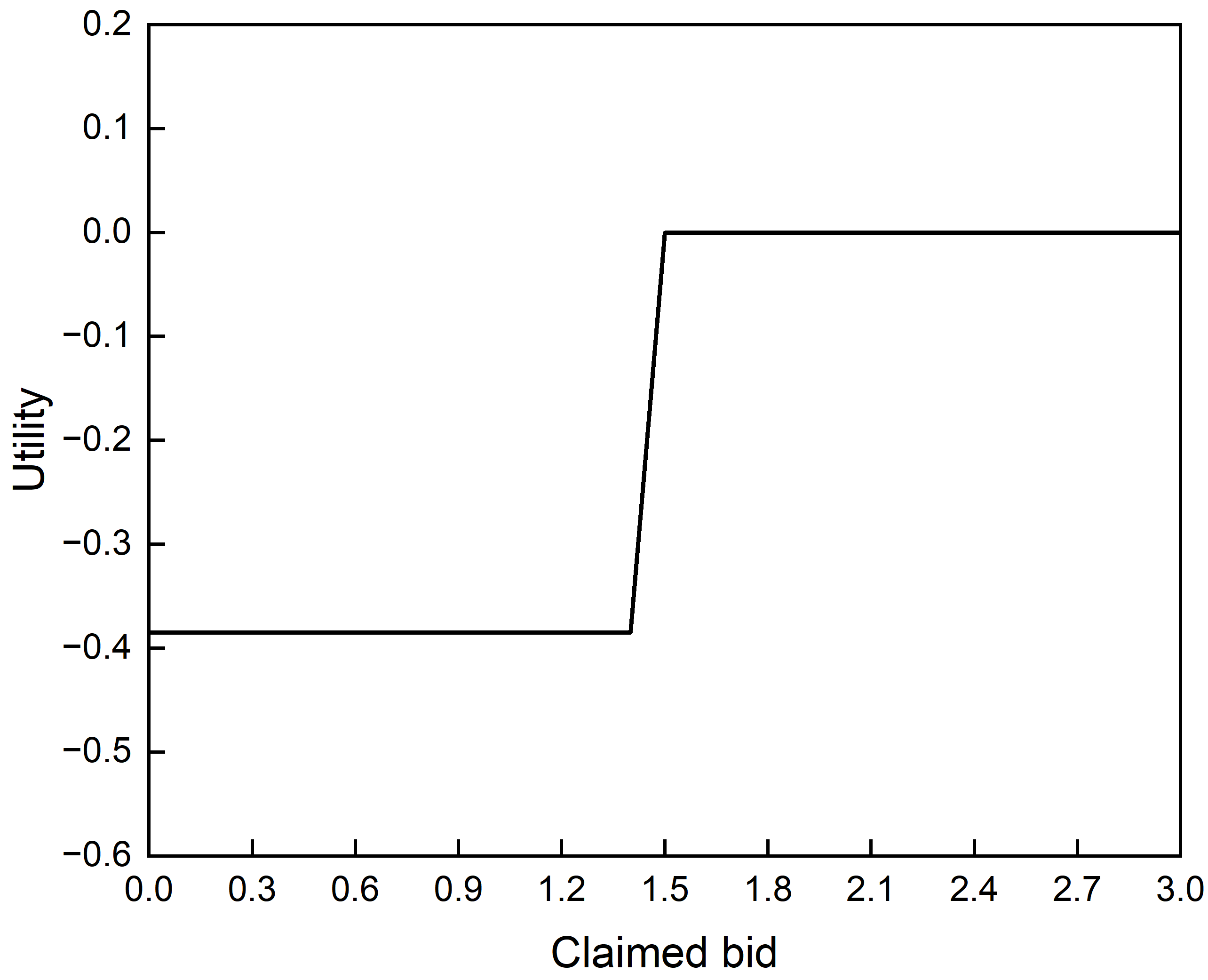}
		%\caption{fig1}
	}%
	\centering
	\caption{A representative user in Dash with $RU=20\%$ and $T=2$.}
	\label{fig8}
\end{figure}

\textit{Truthfulness:} Fig. \ref{fig7} and Fig. \ref{fig8} draw the utility of a representative user under the Damascus and Dash datasets to verify the truthfulness of workers. Shown as Fig. \ref{fig7} (a), when a user $v_i$ gives a truthful bid $b_i=c_i=0.967$, it will win and get utility $u_i=1.61$. If it reduces its bid, its utility will not be changed. If it adds its bid, its utility will be changed to zero when reaching its critical price. Shown as Fig. \ref{fig7} (b), when a user $v_i$ gives a truthful bid $b_i=c_i=1.118$, it will lose and get utility $u_i=0$. If it increases its bid, its utility will not be changed. If it reduces its bid, its utility will be negative. Therefore, the MT-DM-L satisfies the truthfulness according to our experimental verification.

\section{Conclusion \& Future Work}
In this paper, we design an auction-based incentive mechanism to achieve multi-task diffusion maximization for mobile crowdsourcing in social networks. This is totally different from the classical IM problem or single-task crowdsourcing incentive mechanism. We are the first to present the Multi-Task IC model to adapt to information propagation across social networks. Based on it, we formulate a Multi-Task Diffusion Maximization (MT-DM) problem. Even though the Multi-Task Diffusion Function is monotone and submodular, it is \#P-hard to compute. Furthermore, we design a sampling-based algorithm, Modified-OPIM-C, to unbiasedly estimate our objective function, which can return a $(1-1/\sqrt{e}-\varepsilon)$ approximate solution with at least $(1-\delta)$ probability. By combining it with the incentive mechanism, we can achieve individual rationality, truthfulness, and computational efficiency. Our experimental evaluation also validates the effectiveness and correctness of our proposed algorithms, which always output the best performance in an acceptable running time.

\section*{Acknowledgment}

This work was supported in part by the CCF-Huawei Populus Grove Fund under Grant No. CCF-HuaweiLK2022004, the National Natural Science Foundation of China (NSFC) under Grant No. 62202055, the Start-up Fund from Beijing Normal University under Grant No. 310432104, the Start-up Fund from BNU-HKBU United International College under Grant No. UICR0700018-22, and the Project of Young Innovative Talents of Guangdong Education Department under Grant No. 2022KQNCX102.

\ifCLASSOPTIONcaptionsoff
  \newpage
\fi

\bibliographystyle{IEEEtran}
\bibliography{references}

\begin{IEEEbiography}[{\includegraphics[width=1in,height=1.25in,clip,keepaspectratio]{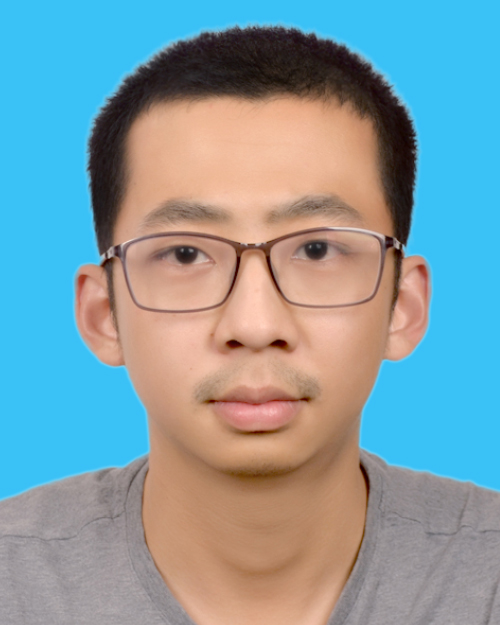}}]{Jianxiong Guo}
	received his Ph.D. degree from the Department of Computer Science, University of Texas at Dallas, Richardson, TX, USA, in 2021, and his B.E. degree from the School of Chemistry and Chemical Engineering, South China University of Technology, Guangzhou, China, in 2015. He is currently an Assistant Professor with the Advanced Institute of Natural Sciences, Beijing Normal University, and also with the Guangdong Key Lab of AI and Multi-Modal Data Processing, BNU-HKBU United International College, Zhuhai, China. He is a member of IEEE/ACM/CCF. He has published more than 40 peer-reviewed papers and been the reviewer for many famous international journals/conferences. His research interests include social networks, wireless sensor networks, combinatorial optimization, and machine learning.
\end{IEEEbiography} 

\begin{IEEEbiography}[{\includegraphics[width=1in,height=1.25in,clip,keepaspectratio]{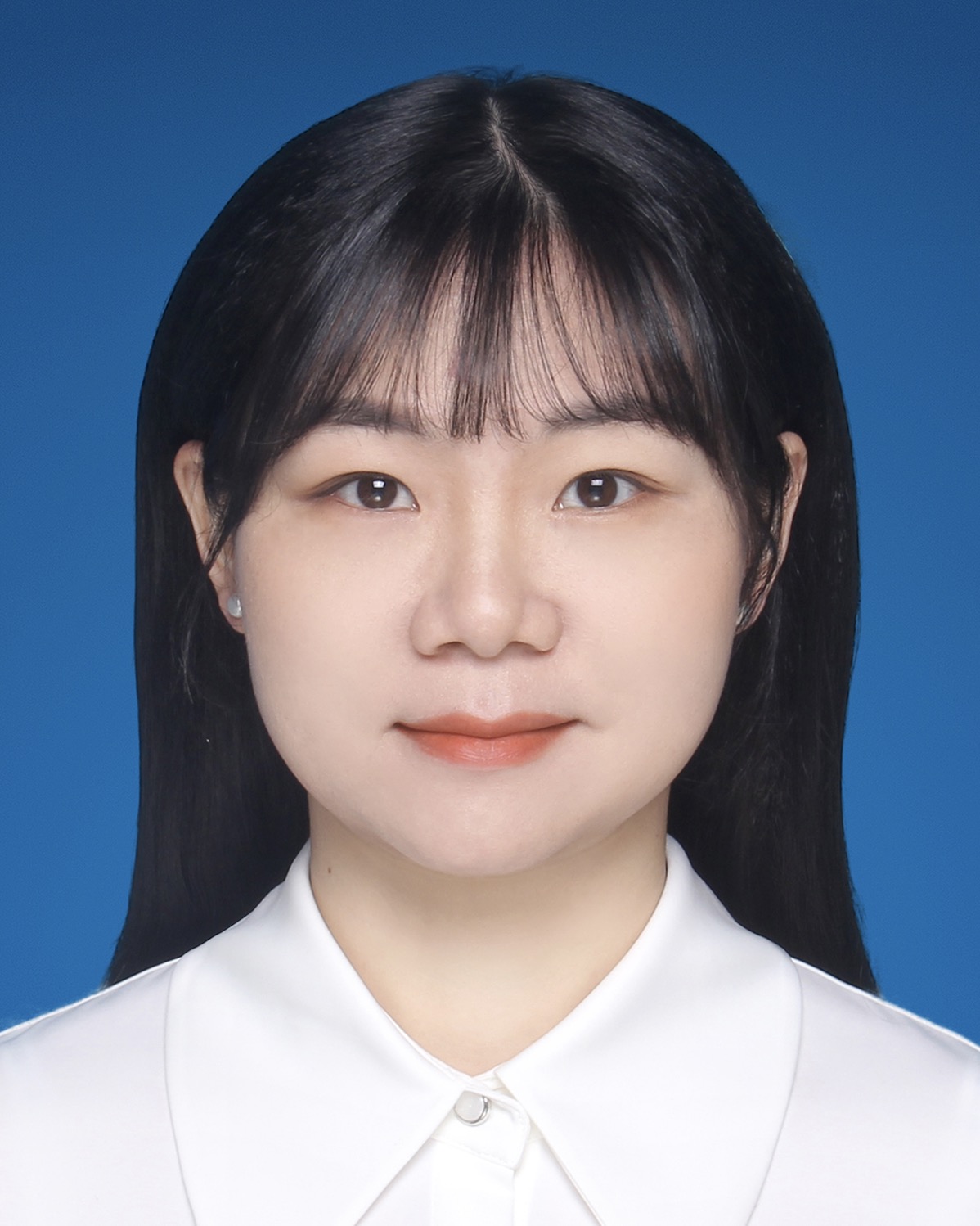}}]{Qiufen Ni}
	is currently an Assistant Professor with the School of Computers, Guangdong University of Technology, Guangzhou, Guangdong 510006, China. She received the Ph.D. degree from School of Computers, Wuhan University in Dec. 2020. During Oct. 2018 to Oct. 2020, she was also a joint Ph.D. student in the Department of Computer Science, University of Texas at Dallas, Richardson, TX, USA. Her research interests include social networks, theoretical approximation algorithm design and analysis, and optimization problems in wireless networks. She serves in China Computer Federation as technical committee member in the Theoretical Computer Science branch committee.
\end{IEEEbiography} 

\begin{IEEEbiography}[{\includegraphics[width=1in,height=1.25in,clip,keepaspectratio]{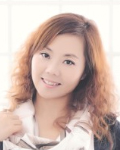}}]{Weili Wu}
	received the Ph.D. and M.S. degrees from the Department of Computer Science, University of Minnesota, Minneapolis, MN, USA, in 2002 and 1998, respectively. She is currently a Full Professor with the Department of Computer Science, The University of Texas at Dallas, Richardson, TX, USA. Her research mainly deals with the general research area of data communication and data management. Her research focuses on the design and analysis of algorithms for optimization problems that occur in wireless networking environments and various database systems.
\end{IEEEbiography}

\begin{IEEEbiography}[{\includegraphics[width=1in,height=1.25in,clip,keepaspectratio]{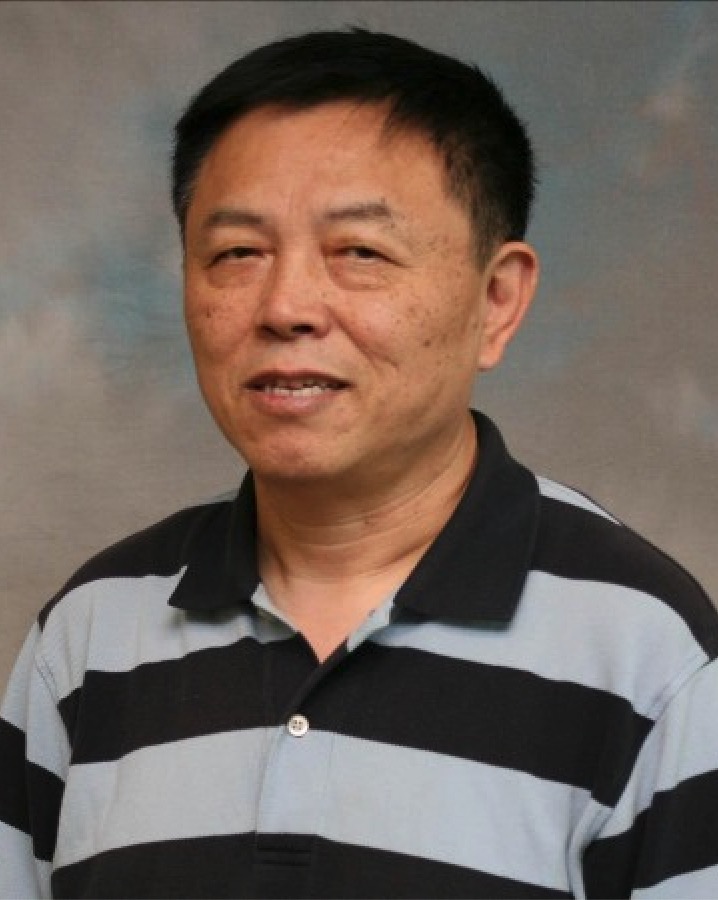}}]{Ding-Zhu Du}
	received the M.S. degree from the Chinese Academy of Sciences, Beijing, China, in 1982, and the Ph.D. degree from the University of California at Santa Barbara, Santa Barbara, CA, USA, in 1985, under the supervision of Prof. R. V. Book. Before settling at The University of Texas at Dallas, Richardson, TX, USA, he was a Professor with the Department of Computer Science and Engineering, University of Minnesota, Minneapolis, MN, USA. He was with the Mathematical Sciences Research Institute, Berkeley, CA, USA, for one year, with the Department of Mathematics, Massachusetts Institute of Technology, Cambridge, MA, USA, for one year, and with the Department of Computer Science, Princeton University, Princeton, NJ, USA, for one and a half years. Dr. Du is the Editor-in-Chief of the Journal of Combinatorial Optimization and is also on the editorial boards for several other journals.
\end{IEEEbiography}
\end{document}